\documentclass[11pt]{article}    

\usepackage[margin=0.75in]{geometry} 
\usepackage{amsmath,amssymb,amsthm}

\newtheorem{theorem}{Theorem}[section]

\newtheorem{lemma}[theorem]{Lemma}

\newtheorem{remark}{Remark}

\newcommand{\al}{\alpha}
\newcommand{\bt}{\beta}
\newcommand{\la}{\lambda}
\newcommand{\ka}{\kappa}

\newcommand{\s}{\sigma}

\newcommand{\be}{\begin{equation}}
\newcommand{\ee}{\end{equation}}
\newcommand{\bea}{\begin{eqnarray}}
\newcommand{\eea}{\end{eqnarray}}
\newcommand{\no}{\nonumber}

\numberwithin{equation}{section}
\begin{document}

\title{Asymptotics of the Smallest Eigenvalue Distributions of\\Freud Unitary Ensembles}
\author{Chao Min\thanks{School of Mathematical Sciences, Huaqiao University, Quanzhou 362021, China; Email: chaomin@hqu.edu.cn}\: and Liwei Wang\thanks{School of Mathematical Sciences, Huaqiao University, Quanzhou 362021, China}}


\date{February 7, 2024}
\maketitle
\begin{abstract}
We consider the smallest eigenvalue distributions of some Freud unitary ensembles, that is, the probabilities that all the eigenvalues of the Hermitian matrices from the ensembles lie in the interval $(t,\infty)$. This problem is related to the Hankel determinants generated by the Freud weights with a jump discontinuity. By using Chen and Ismail's ladder operator approach, we obtain the discrete systems for the recurrence coefficients of the corresponding orthogonal polynomials. This enables us to derive the large $n$ asymptotics of the recurrence coefficients via Dyson's Coulomb fluid approach. We finally obtain the large $n$ asymptotics of the Hankel determinants and that of the probabilities from their relations to the recurrence coefficients and with the aid of some recent results in the literature.
\end{abstract}

$\mathbf{Keywords}$: Random matrix theory; Freud unitary ensembles; Smallest eigenvalue distributions;

Ladder operators; Coulomb fluid; Large $n$ asymptotics.



\section{Introduction}
In random matrix theory (RMT), the joint probability density function for the eigenvalues $\left\{x_j\right\}_{j=1}^n$of $n\times n$
Hermitian matrices from a unitary ensemble is given by \cite{Mehta}
$$
p\left(x_1, x_2, \ldots, x_n\right) \prod_{j=1}^n d x_j=\frac{1}{D_n[w_0]} \frac{1}{n !} \prod_{1 \leq i<j \leq n}(x_j-x_i)^2 \prod_{k=1}^n w_0(x_k) d x_k,
$$
where $ w_0(x)$ is a weight supported on an interval $[a, b]$ ($a$ may be $-\infty$ and/or $b$ may be $\infty$), and all the moments
$$
\mu_j:=\int_a^b x^j w_0(x) d x, \qquad j=0,1,2, \ldots
$$
exist. Here $D_n[w_0]$ is the partition function or the normalization constant
$$
D_n[w_0]=\frac{1}{n !} \int_{[a, b]^n} \prod_{1 \leq i<j \leq n}(x_j-x_i)^2 \prod_{k=1}^n w_0(x_k) d x_k.
$$

A \textit{Freud weight} is a weight function of the form
$
\exp(-Q(x)),\; x\in\mathbb{R},
$
where $Q(x)$ is real, even, nonnegative, and continuously differentiable. Of special interest are the cases $Q(x)=x^{2m},\;m=1,2,3,\ldots$ \cite[Section 18.32]{NIST}.
In this paper, we take
$$
w_0(x):=\mathrm{e}^{-x^{2m}},\qquad m\in\mathbb{Z}^{+},\; x \in \mathbb{R},
$$
which corresponds to the Freud unitary ensemble (Gaussian unitary ensemble for the case $m=1$). We mention that Claeys, Krasovsky and Minakov \cite{CKM} have studied the more general Freud unitary ensemble recently; see also \cite{KM}.

The probability that all the eigenvalues are greater than $t\; (t\in\mathbb{R})$ in the Freud unitary ensemble is given by
\be\label{pnt1}
\mathbb{P}(n, t)  =\frac{\frac{1}{n!} \int_{(t, \infty)^n} \prod_{1 \leq i<j \leq n}(x_j-x_i)^2 \prod_{k=1}^n w_0(x_k) d x_k}{\frac{1}{n !} \int_{(-\infty, \infty)^n} \prod_{1 \leq i<j \leq n}(x_j-x_i)^2 \prod_{k=1}^n w_0(x_k) d x_k}.
\ee
By using Andr\'{e}ief's or Heine's integration formula, the above $n$-fold integrals can be expressed as the Hankel determinants \cite[p. 27]{Szego}, i.e.,
\be\label{pnt}
\mathbb{P}(n, t) = \frac{D_n(t)}{D_n(-\infty)},
\ee
where
$$
D_n(t)=\det\left(\int_t^{\infty} x^{i+j} w_0(x) d x\right)_{i, j=0}^{n-1},
$$
$$
D_n(-\infty)=\det\left(\int_{-\infty}^{\infty} x^{i+j} w_0(x) d x\right)_{i, j=0}^{n-1},
$$
and we have $D_n(-\infty)=D_n[w_0]$. The similar topic on the smallest or largest eigenvalue distributions of the Gaussian, Laguerre and Jacobi random matrix ensembles has been studied extensively, see e.g. \cite{Adler,DK,Forrester,FT,Haine,LCF,LMC,MMM,TW}.

The main goal of this paper is to derive the large $n$ asymptotic expansion (including the higher order terms) of $\mathbb{P}(n, t)$ for fixed $t\in\mathbb{R}$. We mainly consider the $m=1, 2, 3$ cases. It will be seen that our approach can be applied to study the higher $m$ cases. Taking account of (\ref{pnt}) and noting that the large $n$ asymptotics with higher order terms of the denominator $D_n(-\infty)$ has been obtained in \cite{MWC} recently ($D_n(-\infty)$ has the closed-form expression when $m=1$ \cite[p. 321]{Mehta}), thus we need to deal with the numerator $D_n(t)$.

More generally, we consider the Hankel determinant generated by the Freud weight with a jump discontinuity
$$
\mathcal{D}_n(t)=\det\left(\int_{-\infty}^{\infty} x^{i+j} w(x;t) d x\right)_{i, j=0}^{n-1},
$$
and
\begin{equation}\label{weight}
w(x;t):=w_0(x)(A+B \theta(x-t)),\qquad x,\; t \in \mathbb{R},
\end{equation}
where $A \geq 0$, $A+B \geq 0$ and $\theta(x)$ is the Heaviside step function, i.e., $\theta(x)$ is $1$ for $x>0$ and $0$ otherwise. When $A=0, B=1$, we have $\mathcal{D}_n(t)=D_n(t)$. Hankel determinants for the Gaussian, Laguerre and Jacobi weights with jump discontinuities and Fisher-Hartwig singularities have been studied a lot over the past few years, see e.g. \cite{BC2009,BCI,CG,ChenZhang,Its,Lyu,Min2019,Wu1,Wu2}.

In the paper \cite{BC2009}, Basor and Chen considered the Hankel determinant for the general weight with a jump discontinuity, and derived the ladder operators and compatibility conditions for the orthogonal polynomials with this discontinuous weight. We will use their method to analyze $\mathcal{D}_n(t)$.

It is well known that $\mathcal{D}_n(t)$ can be expressed as the product of the square of $L^2$ norms of the corresponding monic orthogonal polynomials (see e.g. \cite[(2.1.6)]{Ismail}). That is,
\be\label{hankel}
\mathcal{D}_n(t)=\prod_{j=0}^{n-1}h_{j}(t),
\ee
where $h_j(t)>0$ is defined by
$$
h_j(t)\delta_{jk}=\int_{-\infty}^{\infty}P_j(x;t)P_k(x;t)w(x;t)dx,\qquad j, k=0,1,2,\ldots.
$$
Here $\delta_{jk}$ denotes the Kronecker delta, and $P_j(x; t)$ are the \textit{monic} polynomials of degree $j$, orthogonal with respect to the weight (\ref{weight}).

An important property of the orthogonal polynomials is the three-term recurrence relation
$$
xP_{n}(x;t)=P_{n+1}(x;t)+\al_n(t) P_n(x;t)+\beta_{n}(t)P_{n-1}(x;t),
$$
with the initial conditions
$$
P_0(x;t):=1,\qquad \beta_0(t) P_{-1}(x;t):=0.
$$
It can be shown that
\be\label{be}
\beta_{n}(t)=\frac{h_{n}(t)}{h_{n-1}(t)}
\ee
and
\begin{equation}\label{alpha1}
\alpha_{n}(t)=\mathrm{p}(n,t)-\mathrm{p}(n+1,t),
\end{equation}
where $\mathrm{p}(n, t)$ is the sub-leading coefficient of the monic polynomial $P_n(x; t)$, i.e.,
$$
P_n\left(x; t\right)=x^n+\mathrm{p}(n, t) x^{n-1}+\cdots,
$$
with the initial value $\mathrm{p}(0, t)=0$.
Moreover, from (\ref{alpha1}) we have the identity
\be\label{sum}
\sum_{j=0}^{n-1}\al_j(t)=-\mathrm{p}(n,t).
\ee
The combination of (\ref{be}) and (\ref{hankel}) gives
\be\label{bd}
\bt_n(t)=\frac{\mathcal{D}_{n+1}(t)\mathcal{D}_{n-1}(t)}{\mathcal{D}_n^2(t)}.
\ee
See \cite{Chihara,Ismail,Szego} for information about orthogonal polynomials.

The ladder operator approach developed by Chen and Ismail \cite{Chen1997} is a very useful and powerful tool to analyze the recurrence coefficients of orthogonal polynomials and the associated Hankel determinants. It has also been successfully applied to solve problems about orthogonal polynomials for the weights with jump discontinuities \cite{BC2009,ChenZhang,Lyu,Min2019}. Following Basor and Chen \cite{BC2009}, our monic orthogonal polynomials with the weight (\ref{weight}) satisfy the lowering and raising operators
$$
\left(\frac{d}{dx}+B_{n}(x)\right)P_{n}(x)=\beta_{n}A_{n}(x)P_{n-1}(x),
$$
$$
\left(\frac{d}{dx}-B_{n}(x)-\mathrm{v}_0'(x)\right)P_{n-1}(x)=-A_{n-1}(x)P_{n}(x),
$$
where $\mathrm{v_0}(x):=-\ln w_0(x)$ and
\be\label{an}
A_{n}(x):=\frac{R_n(t)}{x-t}+\frac{1}{h_{n}}\int_{-\infty}^{\infty}\frac{\mathrm{v}_0'(x)-\mathrm{v}_0'(y)}{x-y}P_{n}^{2}(y)w(y)dy,
\ee
\be\label{bn}
B_{n}(x):=\frac{r_n(t)}{x-t}+\frac{1}{h_{n-1}}\int_{-\infty}^{\infty}\frac{\mathrm{v}_0'(x)-\mathrm{v}_0'(y)}{x-y}P_{n}(y)P_{n-1}(y)w(y)dy.
\ee
Here $R_n(t)$ and $r_n(t)$ are the auxiliary quantities defined by
\be\label{Rnt}
R_n(t):=B \frac{w_0(t)}{h_n(t)}P_n^2(t; t),
\ee
\be\label{rnt}
r_n(t):=B \frac{w_0(t)}{h_{n-1}(t)} P_n(t;t) P_{n-1}(t;t),
\ee
and $P_n(t;t):=P_n(x;t)|_{x=t}$. Note that we often do not display the $t$-dependence of the recurrence coefficients $\al_n, \bt_n$ and the quantity $h_n$ for brevity.

From the definitions of the functions $A_n(x)$ and $B_n(x)$ and taking advantage of the three-term recurrence relation, it was proved that
$A_n(x)$ and $B_n(x)$ satisfy the following compatibility conditions \cite{BC2009}:
\be
B_{n+1}(x)+B_{n}(x)=(x-\alpha_n) A_{n}(x)-\mathrm{v}_0'(x), \tag{$S_{1}$}
\ee
\be
1+(x-\alpha_n)(B_{n+1}(x)-B_{n}(x))=\beta_{n+1}A_{n+1}(x)-\beta_{n}A_{n-1}(x), \tag{$S_{2}$}
\ee
\be
B_{n}^{2}(x)+\mathrm{v}_0'(x)B_{n}(x)+\sum_{j=0}^{n-1}A_{j}(x)=\beta_{n}A_{n}(x)A_{n-1}(x). \tag{$S_{2}'$}
\ee
Here ($S_{2}'$) is obtained from the combination of ($S_{1}$) and ($S_{2}$), and is usually more useful compared with ($S_{2}$) in practice. The compatibility conditions have the information needed to determine the recurrence coefficients $\alpha_n$, $\bt_n$ and other auxiliary quantities.

From the orthogonality we have
$$
\int_{-\infty}^{\infty}P_n(x;t)P_{n-1}(x;t)w_0(x)(A+B\theta(x-t)) dx=0.
$$
Taking a derivative with respect to $t$ gives
\be\label{eq1}
\frac{d}{dt}\mathrm{p}(n,t)=r_n(t).
\ee
Using (\ref{alpha1}), it follows that
\be\label{eq2}
\alpha_n'(t)=r_n(t)-r_{n+1}(t).
\ee

On the other hand, taking a derivative with respect to $t$ in the equality
$$
h_n(t)=\int_{-\infty}^{\infty}P_n^2(x;t)w_0(x)(A+B\theta(x-t)) dx,
$$
we find
\be\label{dlnhnt}
\frac{d}{dt}\ln h_n(t)=-R_n(t).
\ee
It follows from (\ref{be}) that
\be\label{eq3}
\beta_{n}'(t)=\beta_n \left(R_{n-1}(t)- R_n(t)\right).
\ee
Moreover, the combination of (\ref{hankel}) and (\ref{dlnhnt}) gives
\be\label{eq4}
\frac{d}{dt}\ln \mathcal{D}_n(t)=-\sum_{j=0}^{n-1}R_j(t).
\ee
\begin{remark}
It can be seen that the identities (\ref{eq1}), (\ref{eq2}), (\ref{dlnhnt}), (\ref{eq3}) and (\ref{eq4}) hold for the general smooth weight function $w_0(x)$ on $[a, b]$ with a jump at $t\in[a, b]$.
\end{remark}
For the $A=0,\; B=1$ case, recall that
$$
D_n(t)=\frac{1}{n!} \int_{(t, \infty)^n} \prod_{1 \leq i<j \leq n}(x_j-x_i)^2 \prod_{k=1}^n w_0(x_k) d x_k.
$$
This can be viewed as the partition function for the corresponding unitary random matrix ensemble on $(t, \infty)^n$.
From Dyson's Coulomb fluid approach \cite{Dyson}, the eigenvalues (particles) from the ensemble can be approximated as a continuous fluid with an equilibrium density $\s(x)$ supported on $J$, a subset of $\mathbb{R}$, for sufficiently large $n$. Since the potential $\mathrm{v_0}(x)=-\ln w_0(x)=x^{2m}\; (m\in\mathbb{Z}^{+})$ is convex, $J$ is a single interval, say, $(t,b)$.
According to the work of Chen and Ismail \cite{ChenIsmail} (see also \cite{Saff}), the equilibrium density $\sigma(x)$ is determined by minimizing the free energy functional
\be\label{fe}
F[\s]:=\int_{t}^{b}\s(x)\mathrm{v}_0(x)dx-\int_{t}^{b}\int_{t}^{b}\s(x)\ln|x-y|\s(y)dxdy,
\ee
subject to the normalization condition
\be\label{con}
\int_{t}^{b}\s(x)dx=n.
\ee

Following \cite{ChenIsmail}, it is found that the density $\s(x)$ satisfies the integral equation
\be\label{ie}
\mathrm{v_0}(x)-2\int_{t}^{b}\ln|x-y|\s(y)dy=\mathcal{A},\qquad x\in (t,b),
\ee
where $\mathcal{A}$ is the Lagrange multiplier. It is worth noting that $\mathcal{A}$ is a constant independent of $x$ but it depends on $n$ and $t$. From (\ref{fe}) and taking a partial derivative with respect to $n$, we have
\be\label{pd}
\frac{\partial F[\s]}{\partial n}=\mathcal{A},
\ee
where use has been made of (\ref{ie}) and (\ref{con}).

Equation (\ref{ie}) is transformed into the following singular integral equation by taking a derivative with respect to $x$:
\be\label{sie}
\mathrm{v}_0'(x)-2P\int_{t}^{b}\frac{\sigma(y)}{x-y}dy=0,\qquad x\in (t,b),
\ee
where $P$ denotes the Cauchy principal value. Multiplying by $\frac{x}{\sqrt{(b-x)(x-t)}}$ on both sides of (\ref{sie}) and integrating from $t$ to $b$ with respect to $x$, we obtain
\be\label{sup2}
\frac{1}{2\pi}\int_{t}^{b}\frac{x\:\mathrm{v}_0'(x)}{\sqrt{(b-x)(x-t)}}dx=n,
\ee
where use has been made of (\ref{con}) and the important integral formula
$$
P\int_{t}^{b}\frac{1}{(x-y)\sqrt{(b-x)(x-t)}}dx=0,\qquad y\in (t,b).
$$
The endpoint of the support of the equilibrium density, $b$, is determined by equation (\ref{sup2}).
Furthermore, it was shown in \cite{ChenIsmail} that as $n\rightarrow\infty$,
\begin{equation}\label{al}
\alpha_n=\frac{t+b}{2}+O\left(\frac{\partial^2 \mathcal{A}}{\partial t \partial n}\right),
\end{equation}
\begin{equation}\label{beta}
\beta_n=\left(\frac{b-t}{4}\right)^2\left(1+O\left(\frac{\partial^3 \mathcal{A}}{\partial n^3}\right)\right).
\end{equation}

The rest of the paper is organized as follows. In Section 2, we consider the smallest eigenvalue distribution of the Gaussian unitary ensemble, which is the simplest case corresponding to $m=1$. Based on some finite $n$ results about orthogonal polynomials for the Gaussian weight with a jump discontinuity in \cite{Min2019}, we derive the discrete system for the recurrence coefficients of the orthogonal polynomials, and establish the relation between the logarithmic derivative of the Hankel determinant and the recurrence coefficients. By using Dyson's Coulomb fluid approach, we derive the large $n$ asymptotic expansions of the recurrence coefficients with the aid of the discrete system. The large $n$ asyptotics of the Hankel determinant $D_n(t)$ and then that of the probability $\mathbb{P}(n, t)$ are obtained from their relations with the recurrence coefficients.

In Section 3, we study the smallest eigenvalue distribution of the quartic Freud unitary ensemble ($m=2$). Applying the ladder operator approach to the orthogonal polynomials for the quartic Freud weight with a jump discontinuity, we derive the discrete system satisfied by the recurrence coefficients. This enables us to derive the large $n$ asymptotic expansions of the recurrence coefficients by using Dyson's Coulomb fluid approach. The large $n$ asymptotics of $D_n(t)$ and $\mathbb{P}(n, t)$ are obtained from their relations with the recurrence coefficients and the employment of some recent results in the literature. In Section 4, we repeat the development in Section 3 but for the sextic Freud unitary ensemble ($m=3$), which is more complicated.
Finally we give some discussion in Section 5.

\section{Gaussian Unitary Ensemble}
\subsection{Orthogonal Polynomials for the Gaussian Weight with a Jump}
In this section, we consider the simplest case for $m=1$. The weight function $w(x)$ now is
\begin{equation}\label{wei1}
w(x;t)=w_0(x)(A+B \theta(x-t)),\qquad x,\; t \in \mathbb{R},
\end{equation}
where $w_0(x)$ is the Gaussian weight
$$
w_0(x)=\mathrm{e}^{-x^{2}},\qquad x \in \mathbb{R}.
$$

Hankel determinant and orthogonal polynomials for the Gaussian weight with a jump discontinuity in (\ref{wei1}) have been studied by Min and Chen \cite{Min2019}. In the paper \cite{Min2019}, the ladder operator approach has been used to obtain a series of differential and difference equations satisfied by the recurrence coefficients and the auxiliary quantities. It was proved that the diagonal recurrence coefficient $\al_n(t)$ satisfies the fourth Painlev\'{e} equation, and the logarithmic derivative of the Hankel determinant $\mathcal{D}_n(t)$ satisfies the Jimbo-Miwa-Okamoto $\s$-form of the Painlev\'{e} IV equation. The asymptotics of the recurrence coefficients and the logarithmic derivative of the Hankel determinant have also been studied when the jump discontinuity $t$ approaches the soft edge of the Gaussian unitary ensemble. See also \cite{BCI,Its,XZ}.

Following \cite{Min2019}, we have
\be\label{pt1}
\mathrm{v_0}(x)=-\ln w(x_0)=x^2.
\ee
It follows that
\be\label{vp1}
\frac{\mathrm{v}_0'(x)-\mathrm{v}_0'(y)}{x-y}=2.
\ee
Inserting (\ref{vp1}) into the definitions of $A_n(x)$ and $B_n(x)$ in (\ref{an}) and (\ref{bn}), we find
\be\label{anz1}
A_n(x)
=\frac{R_n(t)}{x-t}+2,
\ee
\be\label{bnz1}
B_n(x) =\frac{r_n(t)}{x-t},
\ee
where $R_n(t)$ and $r_n(t)$ are given by (\ref{Rnt}) and (\ref{rnt}), respectively.

Substituting (\ref{anz1}) and (\ref{bnz1}) into ($S_{1}$), we obtain the following two identities:
\be\label{s11}
R_{n}(t)=2\alpha_{n},
\ee
\be\label{s12}
r_{n}(t)+r_{n+1}(t)=(t-\alpha_{n})R_{n}(t).
\ee
Similarly, substituting (\ref{anz1}) and (\ref{bnz1}) into ($S_{2}'$) gives
\be\label{s23}
r_{n}(t)=2\beta_{n}-n,
\ee
\be\label{s21}
r_{n}^{2}(t)=\beta_{n}R_{n-1}(t)R_{n}(t),
\ee
\be\label{s22a}
2tr_{n}(t)+\sum_{j=0}^{n-1}R_{j}(t)=2\beta_{n}\left(R_{n-1}(t)+R_{n}(t)\right).
\ee

Now we are ready to obtain the discrete system for the recurrence coefficients.
\begin{theorem}\label{th1}
The recurrence coefficients $\alpha_n$ and $\bt_n$ satisfy a pair of difference equations
\begin{subequations}\label{ds}
\be\label{d11}
2\bt_n+2\bt_{n+1}-2n-1=2\al_n(t-\al_n),
\ee
\be\label{d12}
(2\bt_n-n)^2=4\al_{n-1}\al_n\bt_n.
\ee
\end{subequations}
\end{theorem}
\begin{proof}
Substituting (\ref{s11}) and (\ref{s23}) into (\ref{s12}) and (\ref{s21}), we obtain (\ref{d11}) and (\ref{d12}), respectively.
\end{proof}
\begin{remark}
In the next subsection, we will see that the above discrete system is crucial to derive the large $n$ asymptotic expansions of $\alpha_n$ and $\bt_n$.
\end{remark}
We also have the Toda equations for the recurrence coefficients in the following theorem.
\begin{theorem}\label{th2}
The recurrence coefficients satisfy the coupled differential-difference equations
\begin{subequations}
\be\label{dd1}
\alpha_{n}'(t)=2(\beta_{n}-\beta_{n+1})+1,
\ee
\be\label{dd2}
\beta_{n}'(t)=2\beta_{n}(\alpha_{n-1}-\alpha_{n}).
\ee
\end{subequations}
\end{theorem}
\begin{proof}
Substituting (\ref{s23}) into (\ref{eq2}), we have (\ref{dd1}). Substituting (\ref{s11}) into (\ref{eq3}), we obtain (\ref{dd2}). See also \cite{Min2019}.
\end{proof}
\begin{remark}
The combination of Theorems \ref{th1} and \ref{th2} will show that $\al_n(t)$ satisfies the Painlev\'{e} IV equation under a simple transformation. See Theorem \ref{thm1} below.
\end{remark}

Next, we give the relation between the logarithmic derivative of the Hankel determinant $\mathcal{D}_n(t)$ and the recurrence coefficients, which will be used to derive the large $n$ asymptotics of $\frac{d}{d t} \ln \mathcal{D}_n(t)$ in the next subsection.
\begin{theorem}
The logarithmic derivative of the Hankel determinant, $\frac{d}{d t} \ln \mathcal{D}_n(t)$, has the following two alternative expressions:
\begin{align}
\frac{d}{d t} \ln \mathcal{D}_n(t)&=2\mathrm{p}(n,t) \label{ex1}\\
&=2t(2\bt_n-n)-4\bt_n(\al_{n-1}+\al_n). \label{ex2}
\end{align}
\end{theorem}
\begin{proof}
From (\ref{eq4}), (\ref{s11}) and (\ref{sum}), we get (\ref{ex1}).
Eliminating $\sum_{j=0}^{n-1}R_{j}(t)$ from (\ref{eq4}) and (\ref{s22a}), we have
\begin{align}
\frac{d}{d t} \ln \mathcal{D}_n(t)&=2tr_n(t)-2\bt_n(R_{n-1}(t)+R_n(t))\no\\
&=2t(2\bt_n-n)-4\bt_n(\al_{n-1}+\al_n),\no
\end{align}
where use has been made of (\ref{s23}) and (\ref{s11}).
\end{proof}
In the end, we show the relation between the recurrence coefficients, the logarithmic derivative of the Hankel determinant $\mathcal{D}_n(t)$ and the Painlev\'{e} equations.
\begin{theorem}\label{thm1}
Let $y(t):=2\al_n(-t)$, then $y(t)$ satisfies the fourth Painlev\'{e} equation \cite{Gromak}
$$
y''(t)=\frac{\left(y'(t)\right)^2}{2y(t)}+\frac{3}{2}y^3(t)+4t y^2(t)+2\left(t^2-\al\right)y(t)+\frac{\bt}{y(t)},
$$
with parameters $\al=2n+1,\; \bt=0$.
The logarithmic derivative of the Hankel determinant, $\s_n(t):=\frac{d}{d t} \ln \mathcal{D}_n(t)$, satisfies the Jimbo-Miwa-Okamoto $\s$-form of the Painlev\'{e} IV equation \cite{Jimbo1981}
$$
\left(\s_n''(t)\right)^2=4\left(t\s_n'(t)-\s_n(t)\right)^2-4\left(\s_n'(t)+\nu_0\right)\left(\s_n'(t)+\nu_1\right)\left(\s_n'(t)+\nu_2\right),
$$
with parameters $\nu_0=\nu_1=0,\; \nu_2=2n$.
\end{theorem}
\begin{proof}
See \cite{Min2019}.
\end{proof}

\subsection{Large $n$ Asymptotics of the Smallest Eigenvalue Distribution of the Gaussian Unitary Ensemble}
In this subsection, we consider the smallest eigenvalue distribution of the Gaussian unitary ensemble. This corresponds to $A=0, B=1$. Note that in this case $\mathcal{D}_n(t)=D_n(t)$.

Substituting (\ref{pt1}) into (\ref{sup2}), we obtain a quadratic equation satisfied by $b$,
$$
\frac{1}{8} \left(3 b^2+2 b t+3 t^2\right)=n.
$$
Since $b>t$ for large $n$, the above equation has a unique solution given by
$$
b=-\frac{t}{3}+\frac{2\sqrt{2(3n-t^2)}}{3}.
$$
It follows that
\be\label{e1}
\frac{t+b}{2}=\sqrt{\frac{2n}{3}}+\frac{t}{3}-\frac{ t^2}{3 \sqrt{6n}}-\frac{ t^4}{36 \sqrt{6}\:n^{3/2}} -\frac{ t^6}{216 \sqrt{6}\:n^{5/2}}-\frac{5  t^8}{5184 \sqrt{6}\:n^{7/2}}+O(n^{-9 / 2}),
\ee
\be\label{e2}
\left(\frac{b-t}{4}\right)^2=\frac{n}{6}-\frac{t}{3} \sqrt{\frac{2n}{3}}+\frac{t^2}{18}+\frac{ t^3}{9 \sqrt{6n}}+\frac{ t^5}{108 \sqrt{6}\:n^{3/2}}+\frac{ t^7}{648 \sqrt{6}\:n^{5/2}}+\frac{5  t^9}{15552 \sqrt{6}\:n^{7/2}}+O(n^{-9/2}).
\ee
Multiplying by $\frac{1}{\sqrt{(b-x)(x-t)}}$ on both sides of (\ref{ie}) and integrating with respect
to $x$ from $t$ to $b$, and following the similar computations in \cite[p. 406--407]{Min2021}, we find the Lagrange multiplier $\mathcal{A}$ is
$$
\mathcal{A}=\frac{1}{\pi}\int_{t}^{b}\frac{\mathrm{v}_0(x)}{\sqrt{(b-x)(x-t)}}dx-2n\ln\frac{b-t}{4}=n-2 n \ln \frac{b-t}{4}=n-2 n \ln \left[\frac{1}{6} \left(\sqrt{2(3n-t^2)}-2 t\right)\right].
$$
The large $n$ series expansion reads
\begin{align}\label{a}
\mathcal{A}=&-n\ln n+ (1 + \ln 6) n+2t \sqrt{\frac{2n}{3}}+t^2+\frac{7t^3}{9} \sqrt{\frac{2}{3n}}+\frac{t^4}{2 n}+\frac{29  t^5}{30 \sqrt{6}\:n^{3/2}}+\frac{t^6}{3 n^2}\no\\ &+\frac{527  t^7}{756 \sqrt{6}\:n^{5/2}}+\frac{t^8}{4 n^3}+O(n^{-7/2}).
\end{align}

With these ingredients in hand, we are now ready to derive the asymptotic expansions of the recurrence coefficients as $n\rightarrow\infty$.
\begin{theorem}
For fixed $t\in\mathbb{R}$, the recurrence coefficients $\alpha_n$ and $\beta_{n}$ have the following large $n$ expansions:
\be\label{aa}
\alpha_n=\sqrt{\frac{2n}{3}}+\frac{2 t}{3}+\frac{t^2+3}{6 \sqrt{6n} }-\frac{t^4+6 t^2+15}{144 \sqrt{6}\: n^{3/2}}+\frac{t}{36 n^2}+\frac{t^6+9 t^4-117 t^2+81}{1728 \sqrt{6}\: n^{5/2}}+\frac{t (7 t^2-6)}{216 n^3}+O(n^{-7/2}),
\ee
\be\label{ab}
\beta_n=\frac{n}{6}-\frac{t\sqrt{n}}{3 \sqrt{6}}+\frac{t^2}{18}-\frac{t^3}{36 \sqrt{6n} }+\frac{1}{72 n}+\frac{t(t^4-51)}{864 \sqrt{6}\: n^{3/2}}+\frac{7 t^2}{288 n^2}-\frac{t^3 (t^4+549)}{10368 \sqrt{6}\: n^{5/2}}+O(n^{-3}).
\ee
\end{theorem}
\begin{proof}
From (\ref{al}), (\ref{beta}), (\ref{e1}), (\ref{e2}) and (\ref{a}), it can be seen that $\al_n$ and $\bt_{n}$ have the large $n$ expansions of the forms
\be\label{al1}
\alpha_n=a_{-1} \sqrt{n}+a_0+\sum_{j=1}^{\infty}\frac{a_{j}}{n^{j/2}},
\ee
\be\label{be1}
\bt_n=b_{-2}n+b_{-1}\sqrt{n}+b_0+\sum_{j=1}^{\infty}\frac{b_{j}}{n^{j/2}},
\ee
where
$$
a_{-1}=\sqrt{\frac{2}{3}},\qquad\qquad b_{-2}=\frac{1}{6}.
$$
Substituting (\ref{al1}) and (\ref{be1}) into the discrete system satisfied by the recurrence coefficients in (\ref{ds}) and taking a large $n$ limit, we obtain the expansion coefficients $a_j$ and $b_j$ recursively by equating the powers of $n$ on both sides. The first few terms are
$$
\begin{aligned}
&a_0=\frac{2 t}{3},\qquad b_{-1}=-\frac{t}{3 \sqrt{6}};\qquad  a_1=\frac{t^2+3}{6 \sqrt{6} },\qquad  b_0=\frac{t^2}{18};\\
&a_2=0,\qquad b_1=-\frac{t^3}{36 \sqrt{6} };\qquad a_3=-\frac{t^4+6 t^2+15}{144 \sqrt{6}},\qquad b_2=\frac{1}{72};\\
&a_4=\frac{t}{36 },\qquad b_3=\frac{t(t^4-51)}{864 \sqrt{6}};\qquad a_5=\frac{t^6+9 t^4-117 t^2+81}{1728 \sqrt{6}},\qquad b_4=\frac{7 t^2}{288 };\\
&a_6=\frac{t (7 t^2-6)}{216 },\qquad b_5=-\frac{t^3 (t^4+549)}{10368 \sqrt{6}}.
\end{aligned}
$$
This completes the proof.
\end{proof}
\begin{remark}
In recent years, the difference equation method has been applied to derive the large $n$ asymptotic expansions of the recurrence coefficients of many semi-classical orthogonal polynomials; see e.g. \cite{Clarkson1,Clarkson2,Deano,Min2021,Min2023}.
\end{remark}

Based on the large $n$ asymptotics of the recurrence coefficients, we proceed to derive the asymptotics of the logarithmic derivative of the Hankel determinant $D_n(t)$, the Hankel determinant and the probability $\mathbb{P}(n,t)$ in the following.
\begin{theorem}
The logarithmic derivative of the Hankel determinant, $\frac{d}{d t} \ln D_n(t)$, has the following expansion as $n\rightarrow\infty$:
\be\label{dd}
\frac{d}{d t} \ln D_n(t)=-\frac{4}{3}  \sqrt{\frac{2}{3}}\: n^{3/2}-\frac{4n t}{3} -\frac{t^2}{3} \sqrt{\frac{2n}{3}}+\frac{2 t^3}{27}-\frac{ t^4+3}{36 \sqrt{6n}}+\frac{t}{18 n}   +\frac{ t^2 (t^4-153)}{1296 \sqrt{6}\:n^{3/2}}+\frac{7 t^3}{216 n^2}+O(n^{-5/2}).
\ee
\end{theorem}
\begin{proof}
From (\ref{ex2}) we have
$$
\frac{d}{d t} \ln D_n(t)=2t(2\bt_n-n)-4\bt_n(\al_{n-1}+\al_n).
$$
Substituting (\ref{aa}) and (\ref{ab}) into the above, we obtain (\ref{dd}) by taking a large $n$ limit.
\end{proof}
\begin{remark}
From (\ref{ex1}) we get
$$
\frac{d}{d t} \ln D_n(t)=2\mathrm{p}(n,t),
$$
where $\mathrm{p}(n,t)$ is the sub-leading coefficient of the monic orthogonal polynomials.
So, the large $n$ asymptotic expansion of $\mathrm{p}(n,t)$ is given by
$$
\mathrm{p}(n,t)=-\frac{2}{3}  \sqrt{\frac{2}{3}}\: n^{3/2}-\frac{2n t}{3} -\frac{t^2}{3} \sqrt{\frac{n}{6}}+\frac{t^3}{27}-\frac{ t^4+3}{72 \sqrt{6n}}+\frac{t}{36 n}   +\frac{ t^2 (t^4-153)}{2592 \sqrt{6}\:n^{3/2}}+\frac{7 t^3}{432 n^2}+O(n^{-5/2}).
$$
\end{remark}

\begin{theorem}
The Hankel determinant $D_n(t)$ has the following asymptotic expansion as $n\rightarrow\infty$:
\begin{align}\label{dnta}
\ln D_n(t)=&\:\frac{1}{2} n^2 \ln n-\left(\frac{3}{4}+\frac{\ln 6}{2}\right)n^2-\frac{4t}{3} \sqrt{\frac{2}{3}}\: n^{3/2}+\left(\ln(2\pi)-\frac{2}{3}t^2\right)n-\frac{  t^3}{9 }\sqrt{\frac{2n}{3}}-\frac{1}{6} \ln n\no\\
&+\frac{t^4}{54}+2\zeta'(-1)-\frac{\ln 2}{6}+\frac{\ln 3}{8}-\frac{t (t^4+15)}{180 \sqrt{6n} }+\frac{t^2}{36 n}+\frac{t^3 (t^4-357)}{9072 \sqrt{6}\: n^{3/2}}+\frac{140 t^4 -87}{17280 n^2}\no\\
&+O(n^{-5/2}),
\end{align}
where $\zeta'(\cdot)$ is the derivative of the Riemann zeta function.
\end{theorem}
\begin{proof}
From (\ref{pd}) and (\ref{a}), we see that the free energy $F[\s]$ has the large $n$ asymptotic expansion
\begin{align}
F[\s]=&-\frac{1}{2} n^2 \ln n+\left(\frac{3}{4}+\frac{\ln 6}{2}\right)n^2+\frac{4t}{3} \sqrt{\frac{2}{3}}\: n^{3/2}+t^2 n +\frac{14  t^3}{9 }\sqrt{\frac{2n}{3}}+\frac{t^4}{2}  \ln n+C\no\\
&-\frac{29  t^5}{15 \sqrt{6n}}-\frac{t^6}{3 n}-\frac{527  t^7}{1134 \sqrt{6}\:n^{3/2}}-\frac{t^8}{8 n^2}+O(n^{-5/2}),\no
\end{align}
where $C$ is an undetermined constant independent of $n$.
Let
$$
F_n(t):=-\ln D_n(t)
$$
be the ``free energy''. It was shown in \cite{ChenIsmail} that $F_n(t)$ is approximated by the free energy $F[\s]$ in (\ref{fe}) for sufficiently large $n$. It was also pointed out in \cite{ChenIsmail1998} that the approximation is very accurate and effective. So,
we have the following large $n$ expansion form for $F_n(t)$:
\be\label{fna}
F_n(t)=c_{6}n^2\ln n+c_5  \ln n+\sum_{j=-\infty}^{4}c_{j}\:n^{j/2},
\ee
where $c_j,\;j=6, 5, 4, \ldots$ are the expansion coefficients to be determined.

Using (\ref{bd}) we have
\be\label{dc}
\ln\bt_n=2F_n(t)-F_{n+1}(t)-F_{n-1}(t).
\ee
Substituting (\ref{ab}) and (\ref{fna}) into equation (\ref{dc}) and letting $n\rightarrow\infty$, we obtain the expansion coefficients $c_j$ (except $c_2$ and $c_0$) by equating coefficients of powers of $n$ on both sides.
The large $n$ asymptotic expansion for $F_n(t)$ reads
\begin{align}
F_n(t)=&-\frac{1}{2} n^2 \ln n+\left(\frac{3}{4}+\frac{\ln 6}{2}\right)n^2+\frac{4t}{3} \sqrt{\frac{2}{3}}\: n^{3/2}+c_2n+\frac{  t^3}{9 }\sqrt{\frac{2n}{3}}+\frac{1}{6} \ln n+c_0\no\\
&+\frac{t (t^4+15)}{180 \sqrt{6n} }-\frac{t^2}{36 n}-\frac{t^3 (t^4-357)}{9072 \sqrt{6}\: n^{3/2}}-\frac{140 t^4 -87}{17280 n^2}+O(n^{-5/2}),\no
\end{align}
where $c_2$ and $c_0$ are still undetermined. It follows that
\begin{align}\label{dnt1}
\ln D_n(t)=&\frac{1}{2} n^2 \ln n-\left(\frac{3}{4}+\frac{\ln 6}{2}\right)n^2-\frac{4t}{3} \sqrt{\frac{2}{3}}\: n^{3/2}-c_2n-\frac{  t^3}{9 }\sqrt{\frac{2n}{3}}-\frac{1}{6} \ln n-c_0\no\\
&-\frac{t (t^4+15)}{180 \sqrt{6n} }+\frac{t^2}{36 n}+\frac{t^3 (t^4-357)}{9072 \sqrt{6}\: n^{3/2}}+\frac{140 t^4 -87}{17280 n^2}+O(n^{-5/2}).
\end{align}

We proceed to determine the two constants $c_2$ and $c_0$.
Using (\ref{dd}) and taking an integral from $0$ to $t$, we have
\be\label{as1}
\ln \frac{D_n(t)}{D_n(0)}=-\frac{4t}{3}\sqrt{\frac{2}{3}}\:  n^{3/2}-\frac{2}{3}t^2  n-\frac{ t^3}{9}\sqrt{\frac{2n}{3}}+\frac{t^4}{54}-\frac{t (t^4+15)}{180 \sqrt{6n}}+\frac{t^2}{36 n}+\frac{ t^3 (t^4-357)}{9072 \sqrt{6}\:n^{3/2}}+\frac{7 t^4}{864 n^2}+O(n^{-5/2}).
\ee
Based on the results of Dea$\mathrm{\tilde{n}}$o and Simm \cite{Deano}, it was shown in \cite[p. 15283]{Min2023} ($\la=0$) that
\be\label{as2}
\ln D_n(0)=\frac{1}{2} n^2 \ln n-\left(\frac{3}{4}+\frac{\ln 6}{2}\right)n^2+n\ln(2\pi)-\frac{1}{6} \ln n+2\zeta'(-1)-\frac{\ln 2}{6}+\frac{\ln 3}{8}+O(n^{-1}).
\ee
The combination of (\ref{as1}) and (\ref{as2}) gives
\be\label{c2}
-c_2=\ln(2\pi)-\frac{2}{3}t^2,
\ee
\be\label{c0}
-c_0=\frac{t^4}{54}+2\zeta'(-1)-\frac{\ln 2}{6}+\frac{\ln 3}{8}.
\ee
Substituting (\ref{c2}) and (\ref{c0}) into (\ref{dnt1}), we finally arrive at (\ref{dnta}).
\end{proof}
\begin{remark}\label{re}
We provide an alternative derivation of the asymptotics of $D_n(0)$ in (\ref{as2}) by using a recent result of Claeys, Krasovsky and Minakov \cite{CKM}.
Let
$$
\widetilde{D}_n^{(1)}(0):=\det\left(\int_0^{\infty} x^{i+j} \mathrm{e}^{-nx^2} d x\right)_{i, j=0}^{n-1}.
$$
From the identity
$$
\det\left(\int_0^{\infty} x^{i+j} \mathrm{e}^{-nx^2} d x\right)_{i, j=0}^{n-1}=\frac{1}{n!} \int_{(0, \infty)^n} \prod_{1 \leq i<j \leq n}\left(x_j-x_i\right)^2 \prod_{k=1}^n \mathrm{e}^{-nx_k^2} d x_k
$$
and by a simple change of variables, we find
\be\label{dn01}
\widetilde{D}_n^{(1)}(0)=n^{-n^2/2}D_n(0).
\ee
From Proposition 1.10 and (1.22) in \cite{CKM} (see also p. 885--887), we take the special values $\al=0,\;\bt=4$ and let $\mu\rightarrow0^{+}$ to match our case. After some elaborate computations, we obtain
\be\label{dn02}
\ln \widetilde{D}_n^{(1)}(0)=-\left(\frac{3}{4}+\frac{\ln 6}{2}\right)n^2+n\ln(2\pi)-\frac{1}{6} \ln n+2\zeta'(-1)-\frac{\ln 2}{6}+\frac{\ln 3}{8}+o(1).
\ee
The combination of (\ref{dn01}) and (\ref{dn02}) gives
$$
\ln D_n(0)=\frac{1}{2} n^2 \ln n-\left(\frac{3}{4}+\frac{\ln 6}{2}\right)n^2+n\ln(2\pi)-\frac{1}{6} \ln n+2\zeta'(-1)-\frac{\ln 2}{6}+\frac{\ln 3}{8}+o(1),
$$
which coincides with (\ref{as2}).
\end{remark}

In the end, we obtain the following result on the large $n$ asymptotics of the smallest eigenvalue distribution of the Gaussian unitary ensemble.
\begin{theorem}
For fixed $t\in\mathbb{R}$, the probability $\mathbb{P}(n, t)$ in (\ref{pnt1}) for the Gaussian unitary ensemble has the large $n$ asymptotics
\begin{align}
\ln \mathbb{P}(n, t)=&-\frac{n^2\ln 3}{2}-\frac{4t}{3} \sqrt{\frac{2}{3}}\: n^{3/2}-\frac{2 nt^2}{3}-\frac{  t^3}{9 }\sqrt{\frac{2n}{3}}-\frac{1}{12} \ln n+\frac{t^4}{54}+\zeta'(-1)-\frac{\ln 2}{6}+\frac{\ln 3}{8}\no\\[8pt]
&-\frac{t (t^4+15)}{180 \sqrt{6n} }+\frac{t^2}{36 n}+\frac{t^3 (t^4-357)}{9072 \sqrt{6}\: n^{3/2}}+\frac{28 t^4 -3}{3456 n^2}+O(n^{-5/2}),\no
\end{align}
where $\zeta'(\cdot)$ is the derivative of the Riemann zeta function.
\end{theorem}
\begin{proof}
Recall that from (\ref{pnt}) we have
\be\label{pnt21}
\ln \mathbb{P}(n, t)=\ln D_n(t)-\ln D_n(-\infty),
\ee
and the large $n$ asymptotics of $\ln D_n(t)$ is given by (\ref{dnta}). It is well known that $D_n(-\infty)$ has the closed-form expression \cite[p. 321]{Mehta} (see also \cite[p. 302]{Min2019})
$$
D_n(-\infty)=(2\pi)^{\frac{n}{2}}2^{-\frac{n^2}{2}}G(n+1),
$$
where $G(\cdot)$ is the Barnes $G$-function and satisfies the functional equation \cite{Barnes,Voros}
$$
G(z+1)=\Gamma(z)G(z),\qquad\qquad G(1)=1.
$$
It follows that
$$
\ln D_n(-\infty)=-\frac{n^2}{2}\ln 2+\frac{n}{2}\ln(2\pi)+\ln G(n+1).
$$
By using the large $n$ asymptotics of Barnes $G$-function \cite[p. 285]{Barnes}
$$
\ln G(n+1)=\frac{1}{2}n^2\ln n-\frac{3}{4}n^2+\frac{n}{2}\ln (2\pi)-\frac{1}{12}\ln n+\zeta'(-1)-\frac{1}{240n^2}+\frac{1}{1008n^4}+O(n^{-6}),
$$
we obtain as $n\rightarrow \infty$,
\be\label{dnmi}
\ln D_n(-\infty)=\frac{1}{2}n^2\ln n-\left(\frac{3}{4}+\frac{\ln 2}{2}\right)n^2+n\ln(2\pi)-\frac{1}{12}\ln n+\zeta'(-1)-\frac{1}{240n^2}+\frac{1}{1008n^4}+O(n^{-6}).
\ee
Substituting (\ref{dnta}) and (\ref{dnmi}) into (\ref{pnt21}), we establish the theorem.
\end{proof}

\section{Quartic Freud Unitary Ensemble}
\subsection{Orthogonal Polynomials for the Quartic Freud Weight with a Jump}
In this section, we consider the $m=2$ case.
The weight function $w(x)$ now reads
$$
w(x;t)=w_0(x)(A+B \theta(x-t)),\qquad x,\; t \in \mathbb{R},
$$
where $w_0(x)$ is the quartic Freud weight
\be\label{fw}
w_0(x)=\mathrm{e}^{-x^{4}},\qquad x \in \mathbb{R}.
\ee
Orthogonal polynomials with the (generalized) quartic Freud weight have been studied extensively, see e.g. \cite{Nevai,BW,VanAssche2,Magnus,BI,Clarkson4,Clarkson1}. They have played an important role in random matrix theory known as the quartic matrix model, see e.g. \cite{BI,BI2,BGM}.

From (\ref{fw}) we have
\be\label{pt}
\mathrm{v_0}(x)=-\ln w(x_0)=x^4.
\ee
It follows that
\be\label{vp}
\frac{\mathrm{v}_0'(x)-\mathrm{v}_0'(y)}{x-y}=4(x^2+xy+y^2).
\ee
Inserting (\ref{vp}) into the definitions of $A_n(x)$ and $B_n(x)$ in (\ref{an}) and (\ref{bn}), we have
\begin{align}\label{anz}
A_n(x)  &=\frac{R_n(t)}{x-t}+\frac{4}{h_n} \int_{-\infty}^{\infty}\left(x^2+x y+y^2\right) P_n^2(y) w(y) d y\no\\
&=\frac{R_n(t)}{x-t}+4 x^2+4x\alpha_n+R_n^*(t),
\end{align}
\begin{align}\label{bnz}
B_n(x) & =\frac{r_n(t)}{x-t}+\frac{4}{h_{n-1}} \int_{-\infty}^{\infty}\left(x^2+x y+y^2\right) P_n(y) P_{n-1}(y)w(y) d y\no\\
&=\frac{r_n(t)}{x-t}+4x \beta_n+r_n^{*}(t),
\end{align}
where $R_n(t)$ and $r_n(t)$ are given by (\ref{Rnt}) and (\ref{rnt}) respectively, and
$$
R_n^{*}(t):=\frac{4}{h_n} \int_{-\infty}^{\infty} y^2 P_n^2(y) w(y) d y,
$$
$$
r_n^{*}(t):=\frac{4}{h_{n-1}} \int_{-\infty}^{\infty} y^2 P_n(y) P_{n-1}(y) w(y) d y.
$$

Next we use the compatibility conditions satisfied by $A_n(x)$ and $B_n(x)$ to analyze the recurrence coefficients and the auxiliary quantities.
Substituting (\ref{anz}) and (\ref{bnz}) into ($S_1$), we obtain the following three identities:
\begin{equation}\label{s2}
4\left(\beta_n+\beta_{n+1}\right)=R_n^{*}(t)-4 \alpha_n^2,
\end{equation}
\begin{equation}\label{s1}
r_n(t)+r_{n+1}(t)=\left(t-\alpha_n\right) R_n(t),
\end{equation}
\begin{equation}\label{s3}
r_n^{*}(t)+r_{n+1}^{*}(t)=R_n(t)-\alpha_n R_n^{*}(t).
\end{equation}
Similarly, substituting (\ref{anz}) and (\ref{bnz}) into ($S_2'$) produces the following six identities:
\begin{equation}\label{s6}
r_n^{*}(t)=4\left(\alpha_{n-1}+\alpha_n\right) \beta_n,
\end{equation}
\begin{equation}\label{s4}
r_n^2(t)=\beta_n R_{n-1}(t) R_n(t),
\end{equation}
\begin{equation}\label{s7}
4 \beta_n^2+ r_n(t)+ n= \beta_n\left(R_{n-1}^{*}(t)+R_n^{*}(t)+4 \alpha_{n-1} \alpha_n\right) ,
\end{equation}
\begin{equation}\label{s8}
2 \beta_n r_n^{*}(t)+ t r_n(t)+ \sum_{j=0}^{n-1} \alpha_j= \beta_n\left(R_{n-1}(t)+R_{n}(t)+\alpha_n R_{n-1}^{*}(t)+\alpha_{n-1} R_n^{*}(t)\right),
\end{equation}\\[-38pt]
\begin{align}\label{s5}
& 8 t \beta_n r_n(t)+2 r_n(t) r_n^{*}(t)+4t^3 r_n(t)+\sum_{j=0}^{n-1} R_j(t) \no\\
& = \beta_n\left[\left(4 t^2+4 t\alpha_{n-1}+R_{n-1}^{*}(t)\right)R_n(t)+\left(4 t^2+4 t\alpha_{n}+R_{n}^{*}(t)\right)R_{n-1}(t)\right],
\end{align}\\[-45pt]
\begin{align}
& (r_n^{*}(t))^2+4 t^2 r_n(t)+8 \beta_n r_n(t) +\sum_{j=0}^{n-1} R_j^{*}(t)\no\\
&=\beta_n\left[4 t R_{n-1}(t)+4 t R_n(t)+4 \alpha_{n-1} R_n(t)+4 \alpha_n R_{n-1}(t)+R_{n-1}^{*}(t) R_n^{*}(t)\right].\no
\end{align}

From the above identities, we first give the expressions of the auxiliary quantities in terms of the recurrence coefficients in the following lemma.
\begin{lemma}\label{le}
The auxiliary quantities $R_n^{*}(t)$, $r_n^{*}(t)$, $R_n(t)$ and $r_n(t)$  can be expressed in terms of the recurrence coefficients $\alpha_n$ and $\bt_n$ as follows:
\be\label{Rn3}
R_n^*(t)=4\left(\alpha_n^2+\beta_n+\beta_{n+1}\right),
\ee
\be\label{rn3}
r_n^*(t)=4\left(\alpha_{n-1}+\alpha_n\right) \beta_n,
\ee
\be\label{Rn4}
R_n(t)=4\left[\alpha_n^3+\alpha_{n-1} \beta_n+2 \alpha_n\left(\beta_n+\beta_{n+1}\right)+\alpha_{n+1} \beta_{n+1}\right],
\ee
\be\label{rn4}
r_n(t)=4 \beta_n\left(\alpha_{n-1}^2+\alpha_{n-1} \alpha_n+\alpha_n^2+ \beta_{n-1}+\beta_n+\beta_{n+1}\right)-n.
\ee
\end{lemma}
\begin{proof}
From (\ref{s2}) we have (\ref{Rn3}) and we restate (\ref{s6}) as (\ref{rn3}). Substituting (\ref{Rn3}) and (\ref{rn3}) into (\ref{s3}) and (\ref{s7}), we obtain (\ref{Rn4}) and (\ref{rn4}), respectively.
\end{proof}
\begin{remark}
The expressions of $R_n^{*}(t)$ and $r_n^{*}(t)$ in (\ref{Rn3}) and (\ref{rn3}) can also be obtained by using the three-term recurrence relation and the orthogonality from their definitions.
\end{remark}
Now we are ready to derive the discrete system for the recurrence coefficients.
\begin{theorem}\label{the}
The recurrence coefficients $\alpha_n$ and $\bt_n$ satisfy the difference equations
\begin{subequations}\label{de}
\begin{align}\label{chafen1}
& 4 \beta_n\left(\alpha_{n-1}^2+\alpha_{n-1} \alpha_n+\alpha_n^2+\beta_{n-1}+\beta_n+\beta_{n+1}\right)\no\\
&+4 \beta_{n+1}\left(\alpha_n^2+\alpha_n \alpha_{n+1}+\alpha_{n+1}^2+\beta_n+\beta_{n+1}+\beta_{n+2}\right)-2n-1\no\\
&=4\left(t-\alpha_n\right)\left[\alpha_n^3+\alpha_{n-1} \beta_n+2 \alpha_n\left(\beta_n+\beta_{n+1}\right)+\alpha_{n+1} \beta_{n+1}\right],
\end{align}
\begin{align}\label{chafen2}
&\left[4 \beta_n\left(\alpha_{n-1}^2+\alpha_{n-1} \alpha_n+\alpha_n^2+\beta_{n-1}+\beta_n+\beta_{n+1}\right)-n\right]^2\no\\
& =16 \beta_n\left[\alpha_{n-1}^3+\alpha_{n-2} \beta_{n-1}+2 \alpha_{n-1}\left(\beta_{n-1}+\beta_n\right)+\alpha_n \beta_n\right]\no\\
&\times\left[\alpha_n^3+\alpha_{n-1} \beta_n+2 \alpha_n\left(\beta_n+\beta_{n+1}\right)+\alpha_{n+1} \beta_{n+1}\right].
\end{align}
\end{subequations}
\end{theorem}
\begin{proof}
Substituting (\ref{Rn4}) and (\ref{rn4}) into (\ref{s1}), we get (\ref{chafen1}). Similarly, substituting (\ref{Rn4}) and (\ref{rn4}) into (\ref{s4}), we obtain (\ref{chafen2}).
\end{proof}

The above discrete system is vital to derive the large $n$ asymptotic expansions of $\alpha_n$ and $\bt_n$ in the next subsection. Similarly as in the $m=1$ case, we also have the coupled differential-difference equations for the recurrence coefficients.
\begin{theorem}
The recurrence coefficients $\alpha_n$ and  $\beta_{n}$ satisfy the differential-difference equations
$$
\alpha_n'(t)=4 \beta_n\left(\alpha_{n-1}^2+\alpha_{n-1} \alpha_n+\alpha_n^2+ \beta_{n-1}+\beta_n\right)-4 \beta_{n+1}\left(\alpha_{n}^2+\alpha_{n} \alpha_{n+1}+\alpha_{n+1}^2+ \beta_{n+1}+\beta_{n+2}\right)+1,
$$
$$
\beta_n'(t)=4\beta_{n}\left(\alpha_{n-1}^3+\alpha_{n-2} \beta_{n-1}+2 \alpha_{n-1}\beta_{n-1}+\alpha_{n-1} \beta_{n}-\alpha_{n}^3-\alpha_{n} \beta_{n}-2 \alpha_{n}\beta_{n+1}-\alpha_{n+1} \beta_{n+1}\right).
$$
\end{theorem}
\begin{proof}
Substituting (\ref{rn4}) into (\ref{eq2}) and (\ref{Rn4}) into (\ref{eq3}) respectively, we obtain the desired results.
\end{proof}


\subsection{Large $n$ Asymptotics of the Smallest Eigenvalue Distribution of the Quartic Freud Unitary Ensemble}
In this subsection, we consider the smallest eigenvalue distribution of the quartic Freud unitary ensemble. This corresponds to $A=0, B=1$ and note that in this case $\mathcal{D}_n(t)=D_n(t)$.
Similarly as the development in Section 2.2, we will derive the large $n$ asymptotics of the recurrence coefficients, the Hankel determinant $D_n(t)$ and the probability $\mathbb{P}(n,t)$ by using Dyson's Coulomb fluid approach.

We start from substituting (\ref{pt}) into (\ref{sup2}), to obtain a quartic equation satisfied by $b$,
$$
35 b^4 + 20 b^3 t + 18 b^2 t^2 + 20 b t^3 + 35 t^4=64n.
$$
This equation has a unique solution subject to the condition $b>t$, and the large $n$ series expansion reads
\be\label{b}
b=\frac{2 \sqrt{2}\:n^{1 / 4}}{\ka}-\frac{t}{7}-\frac{6\sqrt{2}\:t^2}{7 \ka^3 n^{1/4}}-\frac{24 t^3 }{49 \ka^2 n^{1/2}}-\frac{426\sqrt{2}\: t^4}{1715 \ka\: n^{3/4}}
-\frac{234\sqrt{2}\:t^6}{1715\ka^3n^{5/4}}-\frac{10008 t^7}{84035 \ka^2n^{3/2}}+O(n^{-7 / 4}),
\ee
where $\kappa=\sqrt[4]{35}$.
It follows from (\ref{b}) that
\be\label{dayuea}
\frac{t+b}{2}=\frac{ \sqrt{2}\:n^{1 / 4}}{\ka}+\frac{3t}{7}-\frac{3\sqrt{2}\:t^2}{7 \ka^3n^{1/4}}-\frac{12t^3 }{49 \ka^2n^{1/2}}-\frac{213\sqrt{2}\: t^4}{1715 \ka\: n^{3/4}}
-\frac{117\sqrt{2}\:t^6}{1715\ka^3n^{5/4}}-\frac{5004 t^7}{84035 \ka^2n^{3/2}}+O(n^{-7 / 4}),
\ee
\be\label{dayueb}
\left(\frac{b-t}{4}\right)^2=\frac{n^{1/2}}{2 \ka^2}-\frac{2\sqrt{2}\:t\: n^{1 / 4}}{7 \ka}+\frac{17 t^2}{245}-\frac{177 t^4 }{3430 \ka^2n^{1/2}}+\frac{444 \sqrt{2}\: t^5}{12005\ka\: n^{3/4}}+\frac{414 \sqrt{2}\:  t^7}{84035\ka^3n^{5/4}}+O(n^{-3/2}).
\ee

Similarly as in the $m=1$ case, the Lagrange multiplier $\mathcal{A}$ is
$$
\mathcal{A}=\frac{1}{\pi}\int_{t}^{b}\frac{\mathrm{v}_0(x)}{\sqrt{(b-x)(x-t)}}dx-2n\ln\frac{b-t}{4}=\frac{n}{2}-2n \ln \frac{b-t}{4}.
$$
Using (\ref{b}), it follows that as $n\rightarrow\infty$,
\begin{align}\label{Azhangk}
\mathcal{A}=&-\frac{1}{2}n\ln n+\frac{1}{2} (1 + \ln 140) n+\frac{4 \sqrt{2}\:\ka t\: n^{3/4} }{7}+\frac{46 t^2n^{1/2} }{7 \ka^2}+\frac{232 \sqrt{2}\:t^3 n^{1/4} }{147 \ka}+t^4
\no\\[6pt]
&+\frac{452 \sqrt{2}\:  t^5}{49\ka^3n^{1/4}}+\frac{27166 t^6}{5145 \ka^2n^{1/2}}+\frac{936988 \sqrt{2}\:  t^7}{588245 \ka\: n^{3/4}}+\frac{t^8}{n}+O(n^{-5/4}).
\end{align}

Now we are ready to obtain the first important result on the asymptotic expansions of the recurrence coefficients as $n\rightarrow\infty$.
\begin{theorem}\label{thm}
The recurrence coefficients $\alpha_n$ and $\beta_{n}$ have the following large $n$ expansions:
\begin{align}\label{ashiji1}
\alpha_n=&\:\frac{\sqrt{2}\: n^{1 / 4}}{\ka}+\frac{4 t}{7}+\frac{9 t^2}{14 \sqrt{2}\: \ka^3 n^{1 / 4}}+\frac{\ka^2 t^3}{343 n^{1 / 2}}+\frac{6860+1311 t^4}{27440\sqrt{2}\:\ka\: n^{3 / 4}}-\frac{t^2(2257t^4 +8820)}{109760\sqrt{2}\:\ka^3n^{5 / 4}}\no\\[8pt]
&-\frac{t^3(348t^4 +1715)}{67228\ka^2n^{3/2}}-\frac{1120257 t^8+10792152 t^4+65883440}{602362880 \sqrt{2}\: \ka\: n^{7/4}}+\frac{3 t}{140 n^2}+O(n^{-9/4}),
\end{align}
\begin{align}\label{bshiji1}
\beta_n=&\:\frac{n^{1 / 2}}{2 \ka^2}-\frac{3 t\: n^{1 / 4}}{7 \sqrt{2}\: \ka}+\frac{27 t^2}{490}-\frac{t^3}{28 \sqrt{2}\: \ka^3 n^{1 / 4}}+\frac{6 t^4}{1715 \ka^2 n^{1 / 2}}-\frac{3573t^5}{384160\sqrt{2}\:\ka\: n^{3/4}}+\frac{17937t^7}{10756480\sqrt{2}\:\ka^3n^{5/4}}\no\\[8pt]
&+\frac{17424 t^8+2941225}{47059600 \ka^2 n^{3/2}}+\frac{t \left(2707091 t^8-536479440\right)}{8433080320 \sqrt{2}\: \ka\: n^{7/4}}+\frac{1251 t^2}{156800 n^2}+O(n^{-9 / 4}),
\end{align}
where $\kappa=\sqrt[4]{35}$.
\end{theorem}

\begin{proof}
From (\ref{al}), (\ref{beta}), (\ref{dayuea}), (\ref{dayueb}) and (\ref{Azhangk}), it is easy to see that $\al_n$ and $\bt_{n}$ have the large $n$ expansions of the forms
\be\label{f1}
\alpha_n=a_{-1}n^{1/4}+a_0+\sum_{j=1}^{\infty}\frac{a_{j}}{n^{j/4}},
\ee
\be\label{f2}
\bt_n=b_{-2}n^{1/2}+b_{-1}n^{1/4}+b_0+\sum_{j=1}^{\infty}\frac{b_{j}}{n^{j/4}},
\ee
where
$$
a_{-1}=\frac{\sqrt{2} }{\ka},\qquad\qquad b_{-2}=\frac{1}{2 \ka^2}.
$$
Substituting (\ref{f1}) and (\ref{f2}) into the discrete system satisfied by the recurrence coefficients in (\ref{de}) and taking a large $n$ limit, we obtain the expansion coefficients $a_j$ and $b_j$ recursively by equating the powers of $n$. The theorem is then established.
\end{proof}


\begin{theorem}\label{thm2}
The logarithmic derivative of the Hankel determinant, $\frac{d}{d t} \ln D_n(t)$, has the following expansion as $n\rightarrow\infty$:
\begin{align}\label{dlnDnt}
\frac{d}{d t} \ln D_n(t)  =&-\frac{16 \sqrt{2}\:\ka\: n^{7/4}}{49}-\frac{48 t\: n^{3 / 2}}{7 \ka^2}-\frac{516 \sqrt{2}\: t^2 n^{5/4}}{245 \ka}-\frac{2984 t^3 n}{1715}-\frac{87\ka\: t^4n^{3/4}}{343\sqrt{2}} -\frac{2904 t^5 n^{1 / 2}}{1715\ka^2}\no\\[6pt]
 &-\frac{208497t^6 n^{1/4}}{336140\sqrt{2}\:\ka}+\frac{4740552t^7}{20588575}-\frac{22696293 t^8 -9411920}{37647680 \sqrt{2}\: \ka^3n^{1/4}}-\frac{2 t (664784 t^8+1764735)}{20588575 \ka^2n^{1/2}}\no\\[6pt]
 &+O(n^{-3/4}),
\end{align}
where $\kappa=\sqrt[4]{35}$.
\end{theorem}
\begin{proof}
From (\ref{eq4}) we have
$$
\frac{d}{d t} \ln D_n(t)=-\sum_{j=0}^{n-1} R_j(t).
$$
By making use of (\ref{s5}), we find
\begin{align}
\frac{d}{d t} \ln D_n(t)=&\: 8 t \beta_n r_n(t)+2 r_n(t) r_n^{*}(t)+4t^3 r_n(t)-\beta_n\big[\left(4 t^2+4 t\alpha_{n-1}+R_{n-1}^{*}(t)\right)R_n(t) \no\\
&+\left(4 t^2+4 t\alpha_{n}+R_{n}^{*}(t)\right)R_{n-1}(t)\big].\no
\end{align}
Substituting the expressions of the auxiliary quantities $R_n^{*}(t)$, $r_n^{*}(t)$, $R_n(t)$ and $r_n(t)$ in Lemma \ref{le} into the above, we can express $\frac{d}{d t} \ln D_n(t)$ in terms of the recurrence coefficients $\alpha_n$ and $\bt_n$. By making use of (\ref{ashiji1}) and (\ref{bshiji1}), we obtain (\ref{dlnDnt}).
\end{proof}
\begin{remark}
Substituting (\ref{sum}) into (\ref{s8}), one can express the sub-leading coefficient $\mathrm{p}(n,t)$ in terms of the recurrence coefficients $\al_n$ and $\bt_n$ with the aid of Lemma \ref{le}. Then the large $n$ asymptotics of $\mathrm{p}(n,t)$ can be derived by using (\ref{ashiji1}) and (\ref{bshiji1}). We leave it to the interested reader.
\end{remark}
\begin{theorem}
The Hankel determinant $D_n(t)$ has the large $n$ asymptotics
\begin{align}\label{dnt}
\ln D_n(t)  =&\:\frac{1}{4}n^2\ln n-\left(\frac{3}{8}+\ln \big(\sqrt{2}\:\ka\big)\right)n^2-\frac{16 \sqrt{2}\:\ka\:t\: n^{7/4}}{49}-\frac{24 t^2 n^{3 / 2}}{7 \ka^2}-\frac{172 \sqrt{2}\: t^3 n^{5/4}}{245 \ka} \no\\[5pt]
 &+\left(\ln(2\pi)-\frac{746 t^4 }{1715}\right)n-\frac{87 t^5n^{3/4}}{49\sqrt{2}\:\ka^3}-\frac{484 t^6 n^{1 / 2}}{1715\ka^2}-\frac{208497t^7 n^{1/4}}{2352980\sqrt{2}\:\ka}-\frac{1}{6}\ln n\no\\[5pt]
&+\frac{592569t^8}{20588575}+2\zeta'(-1)-\frac{1}{3}\ln 2+\frac{1}{8}\ln 7-\frac{ t (7565431 t^8-28235760)}{112943040 \sqrt{2}\: \ka^3n^{1/4}}\no\\[6pt]
&-\frac{ t^2 (664784 t^8+8823675)}{102942875 \ka^2n^{1/2}}+O(n^{-3/4}),
\end{align}
where $\kappa=\sqrt[4]{35}$, and $\zeta'(\cdot)$ is the derivative of the Riemann zeta function.
\end{theorem}
\begin{proof}
Similarly as in the $m=1$ case, we have the large $n$ expansion form of $F_n(t)=-\ln D_n(t)$ from (\ref{pd}) and (\ref{Azhangk}) by using Dyson's Coulomb fluid approach:
\be\label{fna1}
F_n(t)=c_{10}n^2\ln n+c_9  \ln n+\sum_{j=-\infty}^{8}c_{j}\:n^{j/4},
\ee
where $c_j,\;j=10, 9, 8, \ldots$ are the expansion coefficients to be determined.

Substituting (\ref{bshiji1}) and (\ref{fna1}) into the identity
$$
\ln\bt_n=2F_n(t)-F_{n+1}(t)-F_{n-1}(t)
$$
and letting $n\rightarrow\infty$, we obtain the expansion coefficients $c_j$ (except $c_4$ and $c_0$) by equating coefficients of powers of $n$ on both sides.
The large $n$ asymptotic expansion for $\ln D_n(t)=-F_n(t)$ reads
\begin{align}
\ln D_n(t)  =&\:\frac{1}{4}n^2\ln n-\left(\frac{3}{8}+\ln \big(\sqrt{2}\:\ka\big)\right)n^2-\frac{16 \sqrt{2}\:\ka\:t\: n^{7/4}}{49}-\frac{24 t^2 n^{3 / 2}}{7 \ka^2}-\frac{172 \sqrt{2}\: t^3 n^{5/4}}{245 \ka} \no\\[5pt]
 &-c_4n-\frac{87 t^5n^{3/4}}{49\sqrt{2}\:\ka^3}-\frac{484 t^6 n^{1 / 2}}{1715\ka^2}-\frac{208497t^7 n^{1/4}}{2352980\sqrt{2}\:\ka}-\frac{1}{6}\ln n-c_0\no\\[6pt]
 &-\frac{ t (7565431 t^8-28235760)}{112943040 \sqrt{2}\: \ka^3n^{1/4}}-\frac{ t^2 (664784 t^8+8823675)}{102942875 \ka^2n^{1/2}}+O(n^{-3/4}).\no
\end{align}

To determine the constants $c_4$ and $c_0$, we integrate (\ref{dlnDnt}) from $0$ to $t$ and have
\begin{align}\label{ra}
\ln \frac{D_n(t)}{D_n(0)}  =&-\frac{16 \sqrt{2}\:\ka\:t\: n^{7/4}}{49}-\frac{24 t^2 n^{3 / 2}}{7 \ka^2}-\frac{172 \sqrt{2}\: t^3 n^{5/4}}{245 \ka}-\frac{746 t^4 n}{1715}-\frac{87 t^5n^{3/4}}{49\sqrt{2}\:\ka^3} -\frac{484 t^6 n^{1 / 2}}{1715\ka^2}\no\\[6pt]
 &-\frac{208497t^7 n^{1/4}}{2352980\sqrt{2}\:\ka}+\frac{592569t^8}{20588575}-\frac{ t (7565431 t^8 -28235760)}{112943040 \sqrt{2}\: \ka^3n^{1/4}}-\frac{ t^2 (664784 t^8+8823675)}{102942875 \ka^2n^{1/2}}\no\\[6pt]
 &+O(n^{-3/4}).
\end{align}
Let
$$
\widetilde{D}_n^{(2)}(0):=\det\left(\int_0^{\infty} x^{i+j} \mathrm{e}^{-nx^4} d x\right)_{i, j=0}^{n-1}.
$$
Similarly as in the $m=1$ case (see Remark \ref{re}), we find
$$
D_n(0)=n^{n^2/4}\widetilde{D}_n^{(2)}(0).
$$
From Proposition 1.10 and (1.22) in \cite{CKM}, we take the special values $\al=0,\;\bt=8$ and let $\mu\rightarrow0^{+}$ to obtain the large $n$ asymptotics of $\widetilde{D}_n^{(2)}(0)$:
$$
\ln \widetilde{D}_n^{(2)}(0)=-\left(\frac{3}{8}+\ln \big(\sqrt{2}\:\ka\big)\right)n^2+n\ln(2\pi)-\frac{1}{6}\ln n+2\zeta'(-1)-\frac{1}{3}\ln 2+\frac{1}{8}\ln 7+o(1).
$$
It follows that
\be\label{dn0a}
\ln D_n(0)=\frac{1}{4}n^2\ln n-\left(\frac{3}{8}+\ln \big(\sqrt{2}\:\ka\big)\right)n^2+n\ln(2\pi)-\frac{1}{6}\ln n+2\zeta'(-1)-\frac{1}{3}\ln 2+\frac{1}{8}\ln 7+o(1).
\ee
The combination of (\ref{ra}) and (\ref{dn0a}) gives
$$
-c_4=\ln(2\pi)-\frac{746 t^4 }{1715},
$$
$$
-c_0=\frac{592569t^8}{20588575}+2\zeta'(-1)-\frac{1}{3}\ln 2+\frac{1}{8}\ln 7.
$$
The theorem is then established.
\end{proof}

Finally, we have the following result on the asymptotics of the smallest eigenvalue distribution of the quartic Freud unitary ensemble.
\begin{theorem}
For fixed $t\in\mathbb{R}$, the probability $\mathbb{P}(n, t)$ in (\ref{pnt1}) for the quartic Freud unitary ensemble has the large $n$ asymptotics
\begin{align}
\ln \mathbb{P}(n, t)  =&\:n^2\ln\frac{\sqrt[4]{3}}{\ka} -\frac{16 \sqrt{2}\:\ka\:t\: n^{7/4}}{49}-\frac{24 t^2 n^{3 / 2}}{7 \ka^2}-\frac{172 \sqrt{2}\: t^3 n^{5/4}}{245 \ka}-\frac{746 t^4 n}{1715}-\frac{87 t^5n^{3/4}}{49\sqrt{2}\:\ka^3} \no\\[6pt]
&-\frac{484 t^6 n^{1 / 2}}{1715\ka^2}-\frac{208497t^7 n^{1/4}}{2352980\sqrt{2}\:\ka}-\frac{1}{12}\ln n+\frac{592569t^8}{20588575}+\zeta'(-1)-\frac{1}{4}\ln 2+\frac{1}{8}\ln 7\no\\[8pt]
&-\frac{ t (7565431 t^8-28235760)}{112943040 \sqrt{2}\: \ka^3n^{1/4}}-\frac{ t^2 (664784 t^8+8823675)}{102942875 \ka^2n^{1/2}}+O(n^{-3/4}),\no
\end{align}
where $\kappa=\sqrt[4]{35}$, and $\zeta'(\cdot)$ is the derivative of the Riemann zeta function.
\end{theorem}
\begin{proof}
From (\ref{pnt}) we have
\be\label{pnt2}
\ln \mathbb{P}(n, t)=\ln D_n(t)-\ln D_n(-\infty).
\ee
The large $n$ asymptotic expansion (with higher order terms) of $\ln D_n(-\infty)$ was recently obtained in \cite[Proposition 2.3]{MWC}:
\begin{align}\label{dnm}
\ln D_n(-\infty)=&\:\frac{1}{4}n^2\ln n -\left(\frac{3}{8}+\frac{1}{4}\ln 12\right)n^2+n \ln (2 \pi)-\frac{1}{12}\ln n +\zeta^{\prime}(-1)-\frac{1}{12}\ln 2\no\\[8pt]
&-\frac{89}{11520 n^2}+\frac{6619}{2322432 n^4}+O(n^{-6}).
\end{align}
Substituting (\ref{dnt}) and (\ref{dnm}) into (\ref{pnt2}) gives the desired result.
\end{proof}

\section{Sextic Freud Unitary Ensemble}
\subsection{Orthogonal Polynomials for the Sextic Freud Weight with a Jump}
In this section, we consider the $m=3$ case.
The weight function $w(x)$ now is
$$
w(x;t)=w_0(x)(A+B \theta(x-t)),\qquad x,\; t \in \mathbb{R},
$$
where $w_0(x)$ is the sextic Freud weight
$$
w_0(x)=\mathrm{e}^{-x^{6}},\qquad x \in \mathbb{R}.
$$
We mention that orthogonal polynomials with the (generalized) sextic Freud weight have been studied in \cite{Freud,Sheen1,Sheen2,VanAssche2,Clarkson2}. They have been found to be of great importance in the study of the continuous limit for the Hermitian matrix model in connection with the non-perturbative theory of two-dimensional quantum gravity \cite{IK,Fokas}.

Since
\be\label{v0}
\mathrm{v}_0(x)=-\ln w(x_0)=x^6,
\ee
we have
\be\label{vp3}
\frac{\mathrm{v}_0'(x)-\mathrm{v}_0'(y)}{x-y}=6 (x^4+x^3 y+x^2 y^2+x y^3+y^4).
\ee
Plugging (\ref{vp3}) into the definitions of $A_n(x)$ and $B_n(x)$ in (\ref{an}) and (\ref{bn}),  we find
\begin{equation}\label{anz3}
A_n(x)  =\frac{R_n(t)}{x-t}+6 x^4+6 x^3\alpha_n+6 x^2(\alpha_n^2+\beta_n+\beta_{n+1})+ x G_n(t) + Q_n(t),
\end{equation}
\begin{equation}\label{bnz3}
B_n(x)  =\frac{r_n(t)}{x-t}+6 x^3\beta_n+6 x^2(\alpha_{n-1}+\alpha_n) \beta_n+ x g_n(t)+ q_n(t),
\end{equation}
where $R_n(t)$ and $r_n(t)$ are given by (\ref{Rnt}) and (\ref{rnt}) respectively, and
$$
G_n(t)=\frac{6}{h_n} \int_{-\infty}^{\infty} y^3 P_n^2(y) w(y) d y,
$$
$$
Q_n(t)=\frac{6}{h_n} \int_{-\infty}^{\infty} y^4 P_n^2(y) w(y) d y,
$$
$$
g_n(t)=\frac{6}{h_{n-1}} \int_{-\infty}^{\infty} y^3 P_n(y) P_{n-1}(y) w(y) d y,
$$
$$
q_n(t)=\frac{6}{h_{n-1}} \int_{-\infty}^{\infty} y^4 P_n(y) P_{n-1}(y) w(y) d y.
$$

Substituting (\ref{anz3}) and (\ref{bnz3}) into ($S_1$), we obtain
\begin{equation}\label{s13}
r_n+r_{n+1}=\left(t-\alpha_n\right) R_n,
\end{equation}
\begin{equation}\label{s233}
R_n=q_n+q_{n+1}+\alpha_n Q_n ,
\end{equation}
\begin{equation}\label{s33}
Q_n=g_n+g_{n+1}+\alpha_n G_n ,
\end{equation}
\begin{equation}\label{s43}
G_n=6\left[\alpha_n^3+\alpha_{n-1} \beta_n+2 \alpha_n(\beta_n+\beta_{n+1})+\alpha_{n+1} \beta_{n+1}\right].
\end{equation}
Similarly, substituting (\ref{anz3}) and (\ref{bnz3}) into ($S_2'$) gives
\begin{equation}\label{s53}
r_n^2=\beta_n R_{n-1} R_n,
\end{equation}
\begin{equation}\label{s73}
g_n=6 \beta_n\left(\alpha_{n-1}^2+\alpha_{n-1} \alpha_n+\alpha_n^2+\beta_{n-1}+\beta_n+\beta_{n+1}\right),
\end{equation}
\begin{equation}\label{s83}
q_n=\beta_n\left[G_{n-1}+G_n+6\alpha_n(\alpha_{n-1}^2 +\beta_{n-1}-\beta_n)+6\alpha_{n-1}(\alpha_n^2-\beta_n+\beta_{n+1})\right],
\end{equation}
\begin{align}\label{s93}
r_n =&\:\beta_n\big[Q_{n-1}+Q_n+\alpha_{n-1}G_n+\alpha_n G_{n-1} -2 g_n+6\alpha_n^2 \beta_{n-1}-12\alpha_{n-1} \alpha_n \beta_n\no \\
&+6\alpha_{n-1}^2(\alpha_n^2+\beta_{n+1}) +6(\beta_{n-1}+\beta_n)(\beta_n+\beta_{n+1})\big]-n,
\end{align}
\begin{align}\label{s63}
&2 r_n\big[3 t^5+6 t^2(t+\alpha_{n-1}+\alpha_n) \beta_n+t g_n+q_n\big]-\beta_nR _ { n } \big[6t^4+ 6 t^2(t\alpha_{n-1}+\alpha_{n-1}^2+\beta_{n-1}+\beta_n)\no\\
&+ t G_{n-1}+Q_{n-1}\big]-\beta_nR_{n-1}\big[6t^4+6 t^2(t \alpha_n+\alpha_n^2+\beta_n+\beta_{n+1})+t G_n+Q_n\big]+\sum_{j=0}^{n-1} R_j=0,
\end{align}
\begin{align}\label{s103}
&t r_n-\beta_n\big[R_{n-1}+R_n+\alpha_{n-1}Q_n+\alpha_n Q_{n-1}-2 q_n -2 g_n(\alpha_{n-1}+\alpha_n)\no \\
&+G_{n-1}(\alpha_n^2+\beta_n+\beta_{n+1})+ G_n(\alpha_{n-1}^2+\beta_{n-1}+\beta_n)\big]+\sum_{j=0}^{n-1} \alpha_j=0.
\end{align}
\begin{remark}
From ($S_2'$) there are another three identities involving the sums $\sum_{j=0}^{n-1}G_j,\; \sum_{j=0}^{n-1}Q_j$ and $\sum_{j=0}^{n-1}(\alpha_j^2+\beta_j+\beta_{j+1})$, respectively. We do not write them down since they will not be used in the following analysis.
\end{remark}
Using the above identities, we have the following theorem similarly as in the $m=1$ and $m=2$ cases, though the results are more complicated.
\begin{theorem}
The recurrence coefficients $\alpha_n$ and $\bt_n$ satisfy the discrete system
\begin{subequations}\label{de3}
\begin{align}
&\beta_n\big[g_{n-1}+g_{n+1}+(\alpha_{n-1}+\al_n)(G_{n-1}+G_{n})+6\alpha_n^2 \beta_{n-1}-12\alpha_{n-1} \alpha_n \beta_n\no\\
&+6\alpha_{n-1}^2(\alpha_n^2+\beta_{n+1}) +6(\beta_{n-1}+\beta_n)(\beta_n+\beta_{n+1})\big]\no\\
&+\beta_{n+1}\big[g_{n}+g_{n+2}+(\alpha_{n}+\al_{n+1})(G_{n}+G_{n+1})+6\alpha_{n+1}^2 \beta_{n}-12\alpha_{n} \alpha_{n+1} \beta_{n+1}\no\\
&+6\alpha_{n}^2(\alpha_{n+1}^2+\beta_{n+2}) +6(\beta_{n}+\beta_{n+1})(\beta_{n+1}+\beta_{n+2})\big]-2n-1\no\\
&=(t-\al_n)\Big\{\beta_n\big[G_{n-1}+G_n+6\alpha_n(\alpha_{n-1}^2 +\beta_{n-1}-\beta_n)+6\alpha_{n-1}(\alpha_n^2-\beta_n+\beta_{n+1})\big]\no\\
&+\beta_{n+1}\big[G_{n}+G_{n+1}+6\alpha_{n+1}(\alpha_{n}^2 +\beta_{n}-\beta_{n+1})+6\alpha_{n}(\alpha_{n+1}^2-\beta_{n+1}+\beta_{n+2})\big]\no\\
&+\al_n(g_n+g_{n+1}+\alpha_n G_n)\Big\},
\end{align}
\begin{align}
&\Big\{\beta_n\big[g_{n-1}+g_{n+1}+(\alpha_{n-1}+\al_n)(G_{n-1}+G_{n})+6\alpha_n^2 \beta_{n-1}-12\alpha_{n-1} \alpha_n \beta_n\no\\
&+6\alpha_{n-1}^2(\alpha_n^2+\beta_{n+1}) +6(\beta_{n-1}+\beta_n)(\beta_n+\beta_{n+1})\big]-n\Big\}^2\no\\
&=\bt_n\Big\{\beta_{n-1}\big[G_{n-2}+G_{n-1}+6\alpha_{n-1}(\alpha_{n-2}^2 +\beta_{n-2}-\beta_{n-1})+6\alpha_{n-2}(\alpha_{n-1}^2-\beta_{n-1}+\beta_{n})\big]\no\\
&+\beta_{n}\big[G_{n-1}+G_{n}+6\alpha_{n}(\alpha_{n-1}^2 +\beta_{n-1}-\beta_{n})+6\alpha_{n-1}(\alpha_{n}^2-\beta_{n}+\beta_{n+1})\big]\no\\
&+\al_{n-1}(g_{n-1}+g_{n}+\alpha_{n-1} G_{n-1})\Big\}\Big\{\beta_n\big[G_{n-1}+G_n+6\alpha_n(\alpha_{n-1}^2 +\beta_{n-1}-\beta_n)\no\\
&+6\alpha_{n-1}(\alpha_n^2-\beta_n+\beta_{n+1})\big]+\beta_{n+1}\big[G_{n}+G_{n+1}+6\alpha_{n+1}(\alpha_{n}^2 +\beta_{n}-\beta_{n+1})\no\\
&+6\alpha_{n}(\alpha_{n+1}^2-\beta_{n+1}+\beta_{n+2})\big]+\al_n(g_n+g_{n+1}+\alpha_n G_n)\Big\},
\end{align}
\end{subequations}
where $G_n$ and $g_n$ are given by (\ref{s43}) and (\ref{s73}), respectively.
\end{theorem}
\begin{proof}
Substituting (\ref{s33}) and (\ref{s83}) into (\ref{s233}) and (\ref{s93}), we have
\begin{align}\label{Rne}
R_n=&\:\beta_n\left[G_{n-1}+G_n+6\alpha_n(\alpha_{n-1}^2 +\beta_{n-1}-\beta_n)+6\alpha_{n-1}(\alpha_n^2-\beta_n+\beta_{n+1})\right]\no\\
&+\beta_{n+1}\left[G_{n}+G_{n+1}+6\alpha_{n+1}(\alpha_{n}^2 +\beta_{n}-\beta_{n+1})+6\alpha_{n}(\alpha_{n+1}^2-\beta_{n+1}+\beta_{n+2})\right]\no\\
&+\al_n(g_n+g_{n+1}+\alpha_n G_n),
\end{align}
\begin{align}\label{rne}
r_n=&\:\beta_n\big[g_{n-1}+g_{n+1}+(\alpha_{n-1}+\al_n)(G_{n-1}+G_{n})+6\alpha_n^2 \beta_{n-1}-12\alpha_{n-1} \alpha_n \beta_n\no\\
&+6\alpha_{n-1}^2(\alpha_n^2+\beta_{n+1}) +6(\beta_{n-1}+\beta_n)(\beta_n+\beta_{n+1})\big]-n.
\end{align}
The theorem is established by substituting (\ref{Rne}) and (\ref{rne}) into (\ref{s13}) and (\ref{s53}), respectively.
\end{proof}
\begin{remark}
The differential-difference equations for the recurrence coefficients can be obtained by substituting (\ref{rne}) and (\ref{Rne}) into (\ref{eq2}) and (\ref{eq3}), respectively.
\end{remark}
\subsection{Large $n$ Asymptotics of the Smallest Eigenvalue Distribution of the Sextic Freud Unitary Ensemble}
In this subsection, we consider the smallest eigenvalue distribution of the sextic Freud unitary ensemble. This corresponds to $A=0, B=1$ and we have $\mathcal{D}_n(t)=D_n(t)$.
Similarly as the development in Sections 2.2 and 3.2, we derive the large $n$ asymptotics of the recurrence coefficients, the Hankel determinant $D_n(t)$ and the probability $\mathbb{P}(n,t)$ by using Dyson's Coulomb fluid approach.

Substituting (\ref{v0}) into (\ref{sup2}), we now obtain a sextic equation satisfied by $b$,
$$
3 \left(231 b^6+126 b^5 t+105 b^4 t^2+100 b^3 t^3+105 b^2 t^4+126 b t^5+231 t^6\right)=1024 n.
$$
This equation has a unique solution subject to the condition $b>t$, and the large $n$ series expansion is
\begin{align}\label{b1}
b=&\:\frac{2^{5/3} n^{1/6}}{3^{1/3} \mu}-\frac{t}{11}-\frac{35\times2^{1/3} t^2}{11\times 3^{2/3}\mu^{5} n^{1/6}}-\frac{260\times 2^{2/3} t^3}{363\times 3^{1/3}\mu^4n^{1/3}}-\frac{2105 t^4}{5324 \mu^3n^{1/2}}-\frac{17420\times2^{1/3} t^5}{43923\times 3^{2/3}\mu^2n^{2/3}}\no\\[8pt]
&-\frac{33986515  t^6}{81169704 \times6^{1/3} \mu n^{5/6}}-\frac{33007175  t^8}{30921792 \times6^{2/3} \mu^5n^{7/6}}-\frac{223856750\times 2^{2/3} t^9}{1004475087 \times3^{1/3}\mu^4n^{4/3} }+O(n^{-3/2}),
\end{align}
where $\mu=\sqrt[6]{77}$.
It follows that
\begin{align}\label{f11}
\frac{t+b}{2}=&\:\frac{2^{2/3} n^{1/6}}{3^{1/3} \mu}+\frac{5t}{11}-\frac{35 t^2}{11\times 6^{2/3}\mu^{5} n^{1/6}}-\frac{130\times 2^{2/3} t^3}{363\times 3^{1/3}\mu^4n^{1/3}}-\frac{2105 t^4}{10648 \mu^3n^{1/2}}-\frac{8710\times2^{1/3} t^5}{43923\times 3^{2/3}\mu^2n^{2/3}}\no\\[8pt]
&-\frac{33986515  t^6}{162339408 \times6^{1/3} \mu n^{5/6}}-\frac{33007175  t^8}{61843584 \times6^{2/3} \mu^5n^{7/6}}+O(n^{-4/3}),
\end{align}
\begin{align}\label{f22}
\left(\frac{b-t}{4}\right)^2=&\:\frac{n^{1/3}}{6^{2/3} \mu^2}-\frac{ 6^{2/3} t n^{1/6}}{11 \mu}+\frac{49 t^2}{726}+\frac{5  t^3}{33\times 6^{2/3} \mu^5n^{1/6}}+\frac{25 t^4}{7986 \times6^{1/3} \mu^4n^{1/3}}-\frac{3415  t^5}{351384 \mu^3 n^{1/2}}\no\\[8pt]
&-\frac{35865\times3^{1/3}   t^6}{1127357\times 2^{2/3}\mu^2n^{2/3}}+\frac{106040245  t^7}{1785733488\times 6^{1/3}  \mu n^{5/6}}+O(n^{-7/6}).
\end{align}
Similarly as in the $m=1$ and $m=2$ cases, the Lagrange multiplier $\mathcal{A}$ now is
$$
\mathcal{A}=\frac{1}{\pi}\int_{t}^{b}\frac{\mathrm{v}_0(x)}{\sqrt{(b-x)(x-t)}}dx-2n\ln\frac{b-t}{4}=\frac{n}{3}- 2 n\ln \frac{b-t}{4}.
$$
Using (\ref{b1}), it follows that as $n\rightarrow\infty$,
\begin{align}\label{a3}
\mathcal{A}=&-\frac{1}{3} n \ln n+\frac{1+\ln 2772}{3}n+\frac{ 7\times6^{4/3}  t n^{5/6}}{\mu^5}+\frac{59\mu^2 t^2 n^{2 / 3}}{121\times6^{1/3}}+\frac{2843  t^3n^{1/2}}{363 \mu^3}
+\frac{30557  t^4n^{1/3}}{2662\times 6^{2/3} \mu^2}\no\\[8pt]
&+\frac{2751989  t^5n^{1/6}}{878460 \times6^{1/3} \mu}+t^6+\frac{42796927  t^7}{351384\times6^{2/3} \mu^5n^{1/6}}-\frac{155704532099 t^8}{127552392 \times6^{1/3} \mu^4n^{1 / 3}}+O(n^{-1/2}).
\end{align}
\begin{theorem}
The recurrence coefficients $\alpha_n$ and $\beta_{n}$ have the following large $n$ asymptotic expansions:
\begin{align}\label{alp1}
\alpha_n=&\:\frac{ 2^{2 / 3} n^{1 / 6}}{3^{1 / 3} \mu}+\frac{6t}{11}+\frac{175 t^2}{66\times 6^{2 / 3} \mu^5 n^{1 / 6}}+\frac{25 \times 2^{2 / 3} t^3}{121 \times \times 3^{1 / 3} \mu^4 n^{1 / 3}}+\frac{27235 t^4}{383328  \mu^3 n^{1 / 2}} +\frac{1802\times 2^{1 / 3} t^5}{43923  \times 3^{2 / 3} \mu^2 n^{2 / 3}}\no\\[8pt]
&+\frac{5844218688 +828363115 t^6}{35065312128 \times 6^{1 / 3} \mu n^{5 / 6}}-\frac{25 t^2 (537101669 t^6+2125170432)}{240447854592\times 6^{2/3} \mu^5n^{7/6}}+O(n^{-4/3}),
\end{align}
\begin{align}\label{bet1}
\beta_{n} =&\:\frac{n^{1 / 3}}{6^{2 / 3} \mu^2}-\frac{5 t n^{1 / 6}}{11 \times 6^{1 / 3} \mu}+\frac{125 t^2}{2178}-\frac{25 t^3}{132 \times 6^{2 / 3} \mu^5 n^{1 / 6}}  -\frac{685 t^4}{35937 \times 6^{1 / 3} \mu^4 n^{1 / 3}}-\frac{10847 t^5}{8433216 \mu^3 n^{1 / 2}}\no\\[8pt]
&+\frac{924955 \times 2^{1 / 3} t^6}{273947751 \times 3^{2 / 3}\mu^2 n^{2 / 3}}-\frac{3813104975 t^7}{771436866816 \times 6^{1 / 3} \mu n^{5 / 6}}+\frac{148190928535 t^9}{37028969607168\times6^{2/3} \mu^5 n^{7/6}}\no\\[5pt]
&+O(n^{-4/3}),
\end{align}
where $\mu=\sqrt[6]{77}$.
\end{theorem}
\begin{proof}
From (\ref{al}), (\ref{beta}), (\ref{f11}), (\ref{f22}) and (\ref{a3}), we find that $\al_n$ and $\bt_{n}$ have the large $n$ expansions of the forms
\be\label{alp}
\alpha_n=a_{-1}n^{1/6}+a_0+\sum_{j=1}^{\infty}\frac{a_{j}}{n^{j/6}},
\ee
\be\label{bet}
\bt_n=b_{-2}n^{1/3}+b_{-1}n^{1/6}+b_0+\sum_{j=1}^{\infty}\frac{b_{j}}{n^{j/6}},
\ee
where
$$
a_{-1}=\frac{ 2^{2 / 3}}{3^{1 / 3} \mu},\qquad\qquad b_{-2}=\frac{1}{6^{2 / 3} \mu^2}.
$$
Substituting (\ref{alp}) and (\ref{bet}) into the discrete system (\ref{de3}) and taking a large $n$ limit, we obtain the expansion coefficients $a_j$ and $b_j$ recursively by equating the powers of $n$. This completes the proof.
\end{proof}
\begin{theorem}
The logarithmic derivative of the Hankel determinant, $\frac{d}{d t} \ln D_n(t)$, has the following expansion as $n\rightarrow\infty$:
\begin{align}\label{ddn}
\frac{d}{d t} \ln D_n(t)  =&-\frac{252 \times 6^{1 / 3}  n^{11 / 6}}{11 \mu^5}-\frac{70 \times 6^{2 /3	}  t n^{5 / 3}}{11 \mu^4}-\frac{3775 t^{2} n^{3 / 2}}{363 \mu^3}-\frac{11435 \times 2^{1 / 3} t^{3} n^{4 / 3}}{1331 \times 3^{2 / 3} \mu^2}\no\\[8pt]
&-\frac{4133705 t^{4} n^{7 / 6}}{819896 \times 6^{1 / 3} \mu}-\frac{20575838 t^{5} n}{10146213} -\frac{265978685 t^{6} n^{5 / 6}}{4216608 \times 6^{2 / 3} \mu^5}-\frac{83337950 \times 2^{2 / 3} t^{7} n^{2 / 3}}{13045131 \times 3^{1 / 3} \mu^4}\no\\[8pt]
&-\frac{32420134455 t^{8} n^{1 / 2}}{12698549248 \mu^3}-\frac{7123199480 \times 2^{1 / 3} t^{9} n^{1 / 3}}{4735382553 \times 3^{2 / 3} \mu^2} -\frac{9522701646921115 t^{10} n^{1 / 6}}{16428519515713536 \times 6^{1 / 3} \mu}\no\\[8pt]
&+\frac{146604389147750 t^{11}}{377467340218353}-\frac{5 (12732126804149495465 t^{12}-7589976016259653632)}{13011387456445120512\times 6^{2/3} \mu^5n^{1/6}}\no\\[4pt]
&+O(n^{-1/3}),
\end{align}
where $\mu=\sqrt[6]{77}$.
\end{theorem}
\begin{proof}
From (\ref{eq4}) we have
$$
\frac{d}{d t} \ln D_n(t)=-\sum_{j=0}^{n-1} R_j(t).
$$
By making use of (\ref{s63}), we find
\begin{align}
\frac{d}{d t} \ln D_n(t)=&\:2 r_n\big[3 t^5+6 t^2(t+\alpha_{n-1}+\alpha_n) \beta_n+t g_n+q_n\big]-\beta_nR _ { n } \big[6t^4+ 6 t^2(t\alpha_{n-1}+\alpha_{n-1}^2+\beta_{n-1}+\beta_n)\no\\
&+ t G_{n-1}+Q_{n-1}\big]-\beta_nR_{n-1}\big[6t^4+6 t^2(t \alpha_n+\alpha_n^2+\beta_n+\beta_{n+1})+t G_n+Q_n\big].\no
\end{align}
Substituting (\ref{Rne}), (\ref{rne}), (\ref{s33}) and (\ref{s83}) into the above, we can express $\frac{d}{d t} \ln D_n(t)$ in terms of the recurrence coefficients $\alpha_n$ and $\bt_n$ with the aid of (\ref{s43}) and (\ref{s73}). By making use of (\ref{alp1}) and (\ref{bet1}), we obtain (\ref{ddn}).
\end{proof}
\begin{remark}
Substituting (\ref{sum}) into (\ref{s103}), one can express the sub-leading coefficient $\mathrm{p}(n,t)$ in terms of the recurrence coefficients $\al_n$ and $\bt_n$. Then the large $n$ asymptotics of $\mathrm{p}(n,t)$ can be derived by using (\ref{alp1}) and (\ref{bet1}).
\end{remark}

\begin{theorem}
The Hankel determinant $D_n(t)$ has the large $n$ asymptotics
\begin{align}\label{dnt3}
\ln D_n(t) =& \:\frac{1}{6}n^{2} \ln n-\Big(\frac{1}{4}+ \ln (6^{1/3}\mu)\Big) n^{2}-\frac{252 \times 6^{1 / 3}  t n^{11 / 6}}{11\mu^5}-\frac{35 \times 6^{2 / 3}  t^{2} n^{5 / 3}}{11\mu^4} -\frac{3775 t^{3} n^{3 / 2}}{1089 \mu^3}\no\\[8pt]
&-\frac{11435 t^{4} n^{4 / 3}}{2662 \times 6^{2/3} \mu^2}-\frac{826741 t^{5} n^{7 / 6}}{819896 \times 6^{1 / 3} \mu}+  \Big(\ln(2\pi)-\frac{10287919 t^6 }{30438639}\Big)n-\frac{37996955 t^{7} n^{5/6}}{4216608 \times 6^{2 / 3} \mu^5} \no\\[8pt]
&-\frac{41668975 t^{8} n^{2 / 3}}{26090262 \times 6^{1 / 3} \mu^4} -\frac{10806711485 t^{9} n^{1 / 2}}{38095647744 \mu^3}-\frac{712319948 \times 2^{1 / 3} t^{10} n^{1 / 3}}{4735382553 \times 3^{2 / 3} \mu^2} \no\\[8pt]
&-\frac{9522701646921115 t^{11} n^{1 / 6}}{180713714672848896 \times 6^{1 / 3} \mu}-\frac{1}{6} \ln n+\frac{73302194573875 t^{12}}{2264804041310118}+2\zeta'(-1)\no\\[8pt]
&-\frac{1}{6}\ln 6+\frac{1}{8}\ln 11-\frac{5 t (12732126804149495465 t^{12}-98669688211375497216)}{169148036933786566656\times 6^{2/3}\mu^5 n^{1/6}}\no\\[8pt]
&-\frac{5 t^2 (1593222711395152 t^{12}+661793388694515)}{157242109153816764\times6^{1/3} \mu^4n^{1/3}}+O(n^{-1/2}),
\end{align}
where $\mu=\sqrt[6]{77}$, and $\zeta'(\cdot)$ is the derivative of the Riemann zeta function.
\end{theorem}
\begin{proof}
Similarly as in the $m=1$ and $m=2$ cases, we have the large $n$ expansion form of $F_n(t)=-\ln D_n(t)$ from (\ref{pd}) and (\ref{a3}) by using Dyson's Coulomb fluid approach:
\be\label{fna3}
F_n(t)=c_{14}n^2\ln n+c_{13}  \ln n+\sum_{j=-\infty}^{12}c_{j}\:n^{j/6},
\ee
where $c_j,\;j=14, 13, 12, \ldots$ are the expansion coefficients to be determined.

Substituting (\ref{bet1}) and (\ref{fna3}) into the identity
$$
\ln\bt_n=2F_n(t)-F_{n+1}(t)-F_{n-1}(t)
$$
and letting $n\rightarrow\infty$, we obtain the expansion coefficients $c_j$ (except $c_6$ and $c_0$) by equating coefficients of powers of $n$ on both sides.
The large $n$ asymptotic expansion for $\ln D_n(t)=-F_n(t)$ reads
\begin{align}
\ln D_n(t) =& \:\frac{1}{6}n^{2} \ln n-\Big(\frac{1}{4}+ \ln (6^{1/3}\mu)\Big) n^{2}-\frac{252 \times 6^{1 / 3}  t n^{11 / 6}}{11\mu^5}-\frac{35 \times 6^{2 / 3}  t^{2} n^{5 / 3}}{11\mu^4} -\frac{3775 t^{3} n^{3 / 2}}{1089 \mu^3}\no\\[8pt]
&-\frac{11435 t^{4} n^{4 / 3}}{2662 \times 6^{2/3} \mu^2}-\frac{826741 t^{5} n^{7 / 6}}{819896 \times 6^{1 / 3} \mu}  -c_6n-\frac{37996955 t^{7} n^{5/6}}{4216608 \times 6^{2 / 3} \mu^5}-\frac{41668975 t^{8} n^{2 / 3}}{26090262 \times 6^{1 / 3} \mu^4} \no\\[8pt]
& -\frac{10806711485 t^{9} n^{1 / 2}}{38095647744 \mu^3}-\frac{712319948 \times 2^{1 / 3} t^{10} n^{1 / 3}}{4735382553 \times 3^{2 / 3} \mu^2}-\frac{9522701646921115 t^{11} n^{1 / 6}}{180713714672848896 \times 6^{1 / 3} \mu} \no\\[8pt]
&-\frac{1}{6} \ln n-c_0-\frac{5 t (12732126804149495465 t^{12}-98669688211375497216)}{169148036933786566656\times 6^{2/3}\mu^5 n^{1/6}}\no\\[8pt]
&-\frac{5 t^2 (1593222711395152 t^{12}+661793388694515)}{157242109153816764\times6^{1/3} \mu^4n^{1/3}}+O(n^{-1/2}).\no
\end{align}

On the other hand, integrating (\ref{ddn}) from $0$ to $t$ gives
\begin{align}\label{ra3}
\ln \frac{D_n(t)}{D_n(0)}  =&-\frac{252\times6^{1/3}tn^{11/6}}{11\mu^5}-\frac{35\times6^{2/3}t^2n^{5/3}}{11\mu^4}-\frac{3775 t^3 n^{3/2} }{1089 \mu^3}-\frac{11435 t^4 n^{4/3} }{2662 \times6^{2/3} \mu^2}\no\\[8pt]
&-\frac{826741 t^5 n^{7/6} }{819896 \times6^{1/3}\mu}-\frac{10287919 t^6 n }{30438639}-\frac{37996955 t^7 n^{5/6} }{4216608 \times6^{2/3} \mu^5}
-\frac{41668975 t^8 n^{2/3} }{26090262 \times6^{1/3}\mu^4}\no\\[8pt]
&-\frac{10806711485  t^9n^{1/2}}{38095647744 \mu^3}-\frac{712319948\times2^{1/3}  t^{10}n^{1/3}}{4735382553\times 3^{2/3}\mu^2}-\frac{9522701646921115 t^{11}n^{1/6}}{180713714672848896 \times6^{1/3}\mu}\no\\[8pt]
&+\frac{73302194573875 t^{12}}{2264804041310118}-\frac{5 t (12732126804149495465 t^{12}-98669688211375497216)}{169148036933786566656\times 6^{2/3} \mu^5n^{1/6}}\no\\[4pt]
&+O(n^{-1/3}).
\end{align}
Let
$$
\widetilde{D}_n^{(3)}(0):=\det\left(\int_0^{\infty} x^{i+j} \mathrm{e}^{-nx^6} d x\right)_{i, j=0}^{n-1}.
$$
We have
$$
D_n(0)=n^{n^2/6}\widetilde{D}_n^{(3)}(0).
$$
From Proposition 1.10 and (1.22) in \cite{CKM}, we obtain the following large $n$ asymptotics of $\widetilde{D}_n^{(3)}(0)$ by taking the special values $\al=0,\;\bt=12$ and letting $\mu\rightarrow0^{+}$:
$$
\ln \widetilde{D}_n^{(3)}(0)=-\Big(\frac{1}{4}+\ln (6^{1/3}\mu)\Big)n^2+n\ln(2\pi)-\frac{1}{6}\ln n+2\zeta'(-1)-\frac{1}{6}\ln 6+\frac{1}{8}\ln 11+o(1).
$$
It follows that
\be\label{dn0a3}
\ln D_n(0)=\frac{1}{6}n^2\ln n-\Big(\frac{1}{4}+\ln (6^{1/3}\mu)\Big)n^2+n\ln(2\pi)-\frac{1}{6}\ln n+2\zeta'(-1)-\frac{1}{6}\ln 6+\frac{1}{8}\ln 11+o(1).
\ee
The combination of (\ref{ra3}) and (\ref{dn0a3}) gives
$$
-c_6=\ln(2\pi)-\frac{10287919 t^6 }{30438639},
$$
$$
-c_0=\frac{73302194573875 t^{12}}{2264804041310118}+2\zeta'(-1)-\frac{1}{6}\ln 6+\frac{1}{8}\ln 11.
$$
The theorem is then established.
\end{proof}

Finally, we have the large $n$ asymptotics of the smallest eigenvalue distribution of the sextic Freud unitary ensemble.
\begin{theorem}
For fixed $t\in\mathbb{R}$, the probability $\mathbb{P}(n, t)$ in (\ref{pnt1}) for the sextic Freud unitary ensemble has the large $n$ asymptotics
\begin{align}
\ln \mathbb{P}(n, t)  =&\:n^2\ln\frac{5^{1/6}}{3^{1/6}\mu} -\frac{252 \times 6^{1 / 3}  t n^{11 / 6}}{11\mu^5}-\frac{35 \times 6^{2 / 3}  t^{2} n^{5 / 3}}{11\mu^4} -\frac{3775 t^{3} n^{3 / 2}}{1089 \mu^3}-\frac{11435 t^{4} n^{4 / 3}}{2662 \times 6^{2/3} \mu^2}\no\\[8pt]
&-\frac{826741 t^{5} n^{7 / 6}}{819896 \times 6^{1 / 3} \mu}  -\frac{10287919 t^6 n}{30438639}-\frac{37996955 t^{7} n^{5/6}}{4216608 \times 6^{2 / 3} \mu^5} -\frac{41668975 t^{8} n^{2 / 3}}{26090262 \times 6^{1 / 3} \mu^4}\no\\[8pt]
& -\frac{10806711485 t^{9} n^{1 / 2}}{38095647744 \mu^3}-\frac{712319948 \times 2^{1 / 3} t^{10} n^{1 / 3}}{4735382553 \times 3^{2 / 3} \mu^2} -\frac{9522701646921115 t^{11} n^{1 / 6}}{180713714672848896 \times 6^{1 / 3} \mu}\no\\[8pt]
&-\frac{1}{12} \ln n+\frac{73302194573875 t^{12}}{2264804041310118}+\zeta'(-1)-\frac{1}{6}\ln 2-\frac{1}{12}\ln 3+\frac{1}{8}\ln 11\no\\[8pt]
&-\frac{5 t (12732126804149495465 t^{12}-98669688211375497216)}{169148036933786566656\times 6^{2/3}\mu^5 n^{1/6}}\no\\[8pt]
&-\frac{5 t^2 (1593222711395152 t^{12}+661793388694515)}{157242109153816764\times6^{1/3} \mu^4n^{1/3}}+O(n^{-1/2}),\no
\end{align}
where $\mu=\sqrt[6]{77}$, and $\zeta'(\cdot)$ is the derivative of the Riemann zeta function.
\end{theorem}
\begin{proof}
Recall that
\be\label{pnt3}
\ln \mathbb{P}(n, t)=\ln D_n(t)-\ln D_n(-\infty).
\ee
The large $n$ asymptotic expansion (with higher order terms) of $\ln D_n(-\infty)$ was recently obtained in \cite[Proposition 3.3]{MWC}:
\begin{align}\label{dnm3}
\ln D_n(-\infty)=&\:\frac{1}{6}n^2\ln n -\left(\frac{1}{4}+\frac{1}{6}\ln 60\right) n^2+n \ln (2 \pi)-\frac{1}{12}\ln n +\zeta^{\prime}(-1)-\frac{1}{12}\ln 3\no\\[8pt]
&-\frac{7}{648 n^2}+\frac{8521}{1837080 n^4}+O(n^{-6}).
\end{align}
Substituting (\ref{dnt3}) and (\ref{dnm3}) into (\ref{pnt3}) yields the result.
\end{proof}

\section{Discussion}
In this paper, we have studied the large $n$ asymptotics of the smallest eigenvalue distributions of Freud unitary ensembles. We mainly considered three cases: the Gaussian unitary ensemble ($m=1$), the quartic Freud unitary ensemble ($m=2$) and the sextic Freud unitary ensemble ($m=3$). It can be seen that our method is able to study the higher $m$ cases, although the computations will become more and more complicated for increasing $m$.

In the $m=1$ case, it was found that the recurrence coefficient $\al_n(t)$ satisfies the fourth Painlev\'{e} equation, and the logarithmic derivative of the Hankel determinant $\mathcal{D}_n(t)$ satisfies the Jimbo-Miwa-Okamoto $\s$-form of the Painlev\'{e} IV equation. However, it is not clear if there is any relation between the recurrence coefficients, the Hankel determinants and the Painlev\'{e} equations when $m=2, 3, \ldots.$ This is because the recurrence coefficients and the auxiliary quantities no longer have the simple relations as in the $m=1$ case (see (\ref{s11}) and (\ref{s23})).

In the end, we would like to point out that all the large $n$ asymptotic expansions obtained in this paper can be easily extended up to any higher order following the procedure we present.

\section*{Acknowledgments}
This work was partially supported by the National Natural Science Foundation of China under grant number 12001212, by the Fundamental Research Funds for the Central Universities under grant number ZQN-902 and by the Scientific Research Funds of Huaqiao University under grant number 17BS402.

\section*{Conflict of Interest}
The authors have no competing interests to declare that are relevant to the content of this article.
\section*{Data Availability Statements}
Data sharing not applicable to this article as no datasets were generated or analysed during the current study.

\end{document}